\documentclass[11pt,a4]{article}

\usepackage{microtype}
\usepackage[T1]{fontenc}
\usepackage{amssymb}
\usepackage{amsmath}
\usepackage{amsfonts}
\usepackage{amsthm}
\usepackage{mathtools}
\usepackage{thmtools}
\usepackage{fullpage}
\usepackage{authblk}
\usepackage{thm-restate}
\usepackage[utf8]{inputenc}
\usepackage{hyperref}

\newcommand{\eps}{\varepsilon}

\newcommand{\HAM}{\mathsf{Ham}}
\newcommand{\sh}{\mathsf{sh}}

\newcommand{\MP}{\mathsf{MP}}
\newcommand{\PM}{\mathsf{P}}
\newcommand{\MI}{\mathsf{MI}}
\newcommand{\minrot}{\mathsf{minrot}}

\newcommand{\shft}{\mathsf{cyc}}
\newcommand{\rot}{\mathsf{rot}_n}
\newcommand{\modulo}{\operatorname{mod}}

\newcommand{\floor}[1]{\left\lfloor{#1}\right\rfloor}

\newcommand{\fnp}{f_{\mathsf{n}}}
\newcommand{\fp}{f_{\mathsf{c}}}
\renewcommand{\P}{\mathsf{C}}
\newcommand{\NP}{\mathsf{N}}

\newcommand{\cmod}{\hspace{-0.002\textwidth}\circlearrowright\hspace{-0.002\textwidth}}

\newcommand{\per}{\mathsf{per}}
\newcommand{\prr}{\mathsf{root}}

\newcommand{\dd}{\mathinner{\ldotp\ldotp}}
\newcommand{\Oh}{O}
\newcommand{\Ohtilde}{\tilde{\Oh}}
\newcommand{\Exp}{\mathbb{E}}

\theoremstyle{plain}
\newtheorem{definition}{Definition}[section]
\newtheorem{lemma}[definition]{Lemma}
\newtheorem{theorem}[definition]{Theorem}
\newtheorem{problem}[definition]{Problem}
\newtheorem{proposition}[definition]{Proposition}

\newtheorem{observation}[definition]{Observation}
\newtheorem{corollary}[definition]{Corollary}

\newtheorem{construction}[definition]{Construction}
\newtheorem{fact}[definition]{Fact}
\newtheorem{claim}[definition]{Claim}
\newtheorem{remark}[definition]{Remark}

\renewcommand{\H}{\mathcal{H}}
\newcommand{\sm}{\setminus}
\newcommand{\sub}{\subseteq}

\renewcommand{\S}{\Pi}
\newcommand{\dec}{\mathtt{dec}}
\newcommand{\sk}{\mathtt{sk}}
\renewcommand{\circ}{\mathtt{circ}}

\title{Improved Circular $k$-Mismatch Sketches}

\author[1]{Shay Golan}
\author[1]{Tomasz Kociumaka}
\author[1]{Tsvi Kopelowitz}
\author[1]{Ely Porat}
\author[2]{Przemys\l{}aw~Uzna\'nski}

\affil[1]{Department of Computer Science, Bar-Ilan University, Ramat Gan, Israel\thanks{Supported in part by ISF grants no.\ 1278/16 and 1926/19, by a BSF grant no.\ 2018364, and by an ERC grant MPM under the EU's Horizon 2020 Research and Innovation Programme (grant no. 683064).}}
\affil[2]{Institute of Computer Science, University of Wrocław, Poland\thanks{Supported by Polish National Science Centre grant 2019/33/B/ST6/00298.}}

\date{}
\begin{document}

\maketitle

\begin{abstract}
The shift distance $\sh(S_1,S_2)$ between two strings $S_1$ and $S_2$ of the same length is defined as the minimum Hamming distance between $S_1$ and any rotation (cyclic shift) of $S_2$. We study the problem of sketching the shift distance, which is the following communication complexity problem: Strings $S_1$ and $S_2$ of length $n$
are given to two identical players (encoders), who independently compute sketches (summaries) $\sk(S_1)$ and $\sk(S_2)$, respectively, so that upon receiving the two sketches, a third player (decoder) is able to compute (or approximate) $\sh(S_1,S_2)$ with high probability.

This paper primarily focuses on the more general $k$-mismatch version of the problem, where the decoder is allowed to declare a failure
if $\sh(S_1,S_2)>k$, where $k$ is a parameter known to all parties.
Andoni et al. (STOC'13) introduced exact circular $k$-mismatch sketches of size $\Ohtilde(k+D(n))$, where $D(n)$ is the number of divisors of $n$. Andoni et al.\ also showed that their sketch size is optimal in the class of linear homomorphic sketches.

We circumvent this lower bound by designing a (non-linear) exact circular $k$-mismatch sketch of size $\Ohtilde(k)$; this size matches communication-complexity lower bounds. We also design $(1\pm \eps)$-approximate circular $k$-mismatch sketch of size $\Ohtilde(\min(\eps^{-2}\sqrt{k}, \eps^{-1.5}\sqrt{n}))$, which improves upon an $\Ohtilde(\eps^{-2}\sqrt{n})$-size sketch of Crouch and McGregor (APPROX'11).
\end{abstract}

\section{Introduction}
The \emph{Hamming distance}~\cite{H50} is a fundamental metric for strings,
and computing the Hamming distances in various settings is a central task in text processing.
The Hamming distance of two length-$n$ strings $S_1$ and $S_2$ is defined
as the number of aligned mismatches between $S_1$ and $S_2$.
In the $k$-mismatch variant~\cite{Abrahamson87,AmirLP04,k-mismatch,GU18,K:1987}, the problem is parameterized by an integer $1\le k \le n$,
and the task is relaxed so that if $\HAM(S_1,S_2)>k$, then instead of computing $\HAM(S_1,S_2)$, the algorithm is only required to report that this is the case, without computing the distance.
Since computing the exact Hamming distance, both in the classic version and the $k$-mismatch version, is challenging under some efficiency constraints, a large body of research~\cite{k-mismatch,Karloff93,KP:15,KopelowitzP18} focused on the approximation version of both problems.
Formally, in the $(1\pm \eps)$-approximation variant of either problem, the problem is parameterized by $\eps > 0$, and whenever the algorithm should report $\HAM(S_1,S_2)$ in the original problem, in the approximation variant, the algorithm may report a $(1\pm\eps)$-approximation of $\HAM(S_1,S_2)$.

\subparagraph*{Sketching.} Sketching is one of the settings of sublinear algorithms designed for space-efficient and time-efficient processing of big data, with applications in streaming algorithms, signal processing, network traffic monitoring, and other areas~\cite{DBLP:journals/cacm/Cormode17,cormode2011sketch,nelson2012sketching}.
The task of sketching the Hamming distance boils down to constructing two (randomized) functions $\sk : \Sigma^n \to \{0,1\}^*$ and $\dec : \{0,1\}^*\times \{0,1\}^* \to \mathbb{N}$
such that $\dec(\sk(S_1),\sk(S_2))=\HAM(S_1,S_2)$ holds with high probability\footnote{An event $\mathcal E$ is said to happen with high probability if $\Pr[\mathcal E] \ge 1-n^{-\Omega(1)}$.}.
The communication-complexity interpretation of this problem involves three players sharing public randomness:
two identical encoders and a decoder.
The first encoder receives a string $S_1$, while the second encoder receives a string $S_2$.
Each of the encoders needs to independently summarize its string.
The summaries (sketches) are then sent to the decoder, whose task is to retrieve $\HAM(S_1,S_2)$
based on the summaries alone, without access to $S_1$ or $S_2$.
The sketching complexity of Hamming distance, which is the size of the sketch, is well understood:
the optimal sketch size is $\tilde{\Theta}(n)$ for the base variant~\cite{PL07,DBLP:conf/soda/Woodruff04}, $\tilde{\Theta}(k)$ for the $k$-mismatch variant~\cite{DBLP:journals/ipl/HuangSZZ06,PL07},
and $\tilde{\Theta}(\eps^{-2})$ for the $(1\pm \eps)$-approximate variants~\cite{ALON1999137,DBLP:journals/siamcomp/KushilevitzOR00,DBLP:conf/soda/Woodruff04}.\footnote{Throughout this paper, the $\tilde\Theta(\cdot),\tilde \Omega(\cdot)$, and $\Ohtilde(\cdot)$ notations suppress $\log^{\Oh(1)}n$ factors.}
Much less is known about the sketching complexity of edit distance: it is $\tilde{\Theta}(n)$
for the base variant and $\Ohtilde(k^8)$ for the $k$-error variant~\cite{BZ16}.
Approximate edit distance sketches with super-constant approximation ratios are also known; see e.g.~\cite{CGK16,OR07}.

\subparagraph*{The shift distance.}  We consider the \emph{shift distance}~\cite{AGMP13,AIK08,CM11}, which is a cyclic variant of Hamming distance.
For two strings $S_1,S_2\in \Sigma^n$, the shift distance is defined as the minimum
Hamming distance between $S_1$ and any cyclic shift (rotation) of $S_2$.
Formally, if $\shft$ is a function cyclically shifting a given string (by one position to the left), then  $\sh(S_1,S_2)=\min\{\HAM(S_1,\shft^m(S_2)) \mid m \in \mathbb{Z}\}$ is the shift distance between $S_1$ and $S_2$.
The research on shift distance for sublinear algorithms is mostly motivated by the observation that the shift distance shares many similarities with the fundamental  Hamming distance.
At the same time, shift distance  inherits some of the challenges exhibited in the edit distance, e.g., in the context of low-dimensional embeddings to $\ell_1$~\cite{KN06} and asymmetric query complexity~\cite{AKO10}.

The first sketching scheme for shift distance, by Andoni et al.~\cite{AIK08},
allows for $\Oh(\log^2 n)$-approximation using sketches of size $\Ohtilde(1)$.
Crouch and McGregor~\cite{CM11} showed $(1\pm\eps)$-approximate sketches for shift distance that use $\Ohtilde(\eps^{-2}\sqrt{n})$ space.
Andoni et al.~\cite{AGMP13} designed exact $k$-mismatch circular sketches that use $\Ohtilde(D(n)+k)$ space, where $D(n)$ is the number of divisors of $n$, which is $n^{\Theta(1/\log \log n)}$ in the worst case. In~\cite{AGMP13}, it is proven that $\tilde \Omega(D(n))$ is a lower bound for any linear homomorphic sketch for the shift distance $k$-mismatch problem.\footnote{A sketch is \emph{homomorphic} if $\sk(\shft(S))$ can be retrieved from $\sk(S)$ and \emph{linear} if $\sk$ is a linear mapping.}

\subparagraph*{Our results.}
We consider a (slight) generalization of the problem of sketching the shift distance, where the decoder needs to retrieve $\HAM(S_1,\shft^m(S_2))$ for every $m\in \mathbb{Z}$.
We consider the problem both in the exact setting and in the $(1\pm\eps)$-approximation version.

\begin{problem}\label{def:circksk}
	An \emph{exact circular $k$-mismatch sketch} ($k$-ECS) for $\S\sub \Sigma^n$ is a pair of randomized functions\footnote{
	A randomized function $f : X \to Y$ is a random variable whose values are functions from $X$ to $Y$.}
	$\sk : \S \to \{0,1\}^{*}$ and $\dec : \{0,1\}^{*} \times \{0,1\}^{*} \times \mathbb{Z} \to \mathbb{N}$ such that,
	for every $S_1,S_2\in \S$ and $m\in \mathbb{Z}$, the following holds with high probability:
	\begin{itemize}
		\item if $\HAM(S_1,\shft^m(S_2)) \le k$, then $\dec(\sk(S_1),\sk(S_2),m)=\HAM(S_1,\shft^m(S_2))$,
		\item otherwise, $\dec(\sk(S_1),\sk(S_2),m)> k$.
	\end{itemize}
\end{problem}

\begin{problem}\label{def:circepssk}
	A \emph{$(1\pm \eps)$-approximate circular $k$-mismatch sketch} ($(\eps,k)$-ACS) for $\S\sub \Sigma^n$
	is  a pair of randomized functions $\sk : \S \to \{0,1\}^{*}$ and $\dec : \{0,1\}^{*} \times \{0,1\}^{*} \times \mathbb{Z} \to \mathbb{R}$ such that, for every $S_1,S_2\in \S$ and $m\in \mathbb{Z}$,
	the following holds with high probability:
	\begin{itemize}
		\item if $\HAM(S_1,\shft^m(S_2)) \le k$, then $\dec(\sk(S_1),\sk(S_2),m)\in (1\pm\eps)\HAM(S_1,\shft^m(S_2))$,
		\item otherwise, $\dec(\sk(S_1),\sk(S_2),m)>(1-\eps)k$.
	\end{itemize}
\end{problem}

In this paper,  a sketch for $\Pi \subseteq \Sigma^n$ is of size $s$ if for every $S\in \S$, we have $|\sk(S)|\le s$  with high probability.
Our results are stated in the following theorems.

\begin{restatable}{theorem}{thmexact}\label{thm:exact-circular-sketch}
	There exists a $k$-ECS sketch for $\Sigma^n$ of size $\Ohtilde(k)$.
\end{restatable}

\begin{restatable}{theorem}{thmappx}\label{thm:approx_circular_sketch}
	There exists an $(\eps,k)$-ACS sketch for $\Sigma^n$ of size $\Ohtilde(\min(\eps^{-2}\sqrt k,\eps^{-1.5}\sqrt n))$.
\end{restatable}

Notice that Theorem~\ref{thm:exact-circular-sketch} circumvents the lower bound of Andoni et al.~\cite{AGMP13} by using non-linear sketches (however, the sketches are still homomorphic).
Moreover, Theorem~\ref{thm:approx_circular_sketch} improves upon the $\tilde O(\eps^{-2}\sqrt n)$ size sketches of Crouch and McGregor~\cite{CM11}, and also addresses the more general $k$-mismatch variant of the problem.

\subparagraph*{Decoding efficiency.}
We also discuss the efficiency of evaluating $\dec(\sk(S_1),\sk(S_2),m)$ for a given $m\in \mathbb{Z}$ and the efficiency of evaluating or approximating $\sh(S_1,S_2)$ based on our sketches.
We show that the naive solution of minimizing $\dec(\sk(S_1),\sk(S_2),m)$ across all $m\in [n]$
can be sped up significantly. Formally, this yields solutions to the following problems.

\newcommand{\shdec}{\dec^{\sh}}
\begin{problem}\label{def:shiftsk}
	An \emph{exact $k$-mismatch shift distance sketch} ($k$-ESDS) for $\S\sub \Sigma^n$ is a pair of randomized functions
	$\sk : \S \to \{0,1\}^{*}$ and $\shdec : \{0,1\}^{*} \times \{0,1\}^{*} \to \mathbb{N}$ such that,
	for every $S_1,S_2\in \S$, the following holds with high probability:
	\begin{itemize}
		\item if $\sh(S_1,S_2) \le k$, then $\shdec(\sk(S_1),\sk(S_2))=\sh(S_1,S_2)$,
		\item otherwise, $\shdec(\sk(S_1),\sk(S_2))> k$.
	\end{itemize}
\end{problem}

\begin{problem}\label{def:shiftepssk}
	A \emph{$(1\pm \eps)$-approximate $k$-mismatch shift distance sketch} ($(\eps,k)$-ASDS) for $\S\sub \Sigma^n$ is a pair of randomized functions
	$\sk : \S \to \{0,1\}^{*}$ and $\shdec : \{0,1\}^{*} \times \{0,1\}^{*} \to \mathbb{R}$ such that,
	for every $S_1,S_2\in \S$, the following holds with high probability:
	\begin{itemize}
		\item if $\sh(S_1,S_2) \le k$, then $\shdec(\sk(S_1),\sk(S_2))\in (1\pm\eps)\sh(S_1,S_2)$,
		\item otherwise, $\shdec(\sk(S_1),\sk(S_2))>(1-\eps)k$.
	\end{itemize}
\end{problem}

The task of designing efficient algorithms for computing our sketches is left open.

\subparagraph*{Related work.}
A problem closely related to the \emph{circular Hamming distances} problem, asking to determine $\HAM(S_1,\shft^m(S_2))$ for all $0\le m < n$, is the \emph{text-to-pattern Hamming distances} problem, where the input consists of a pattern $P$ (of length $m$) and a text $T$ (of length $n$),
and the task is to compute the Hamming distances between $P$ and every length-$m$ substring of $T$.
A straightforward reduction from the circular Hamming distances problem to the text-to-pattern Hamming distances problem is given by $P=S_1$ and $T=S_2\cdot S_2$.

In the offline setting, including the exact and approximate $k$-mismatch variants, we are not aware of any separation between the two problems.
The state-of-the-art exact solution combines an $\Ohtilde(n\sigma)$-time solution for small alphabets (of size $\sigma$)~\cite{FP:1974}
with an $\Ohtilde(n+\frac{nk}{\sqrt{m}})$-time algorithm~\cite{GU18}, which culminates a long line of research~\cite{Abrahamson87,AmirLP04,k-mismatch,K:1987}.
The approximate variant can be solved in $\Ohtilde(\eps^{-1} n)$ time~\cite{KP:15,KopelowitzP18}; these results improve upon~\cite{Karloff93}.
On the other hand, sketches for text-to-pattern Hamming distances need to be much larger than circular sketches:
already recovering exact occurrences requires $\Omega(n-m)$ space~\cite{DBLP:conf/approx/Bar-YossefJKK04}.

Interestingly, both in the exact and in the approximate setting, the sizes of our circular $k$-mismatch sketches coincide with the current upper bounds for space usage in the \emph{streaming} $k$-mismatch problem.
In that model, the text arrives in a stream, one character at a time, and the goal is to compute, or estimate, after the arrival of each text character, the Hamming distance between $P$ and the current suffix of $T$.
The state-of-the-art exact algorithm~\cite{CKP19} uses $\Ohtilde(k)$ space and costs $\Ohtilde(\sqrt{k})$ time per character, which improves upon~\cite{k-mismatch,DBLP:conf/icalp/GolanKP18,Porat:09,RS}.
A recent approximate streaming algorithm~\cite{approxk} uses $\Ohtilde(\min(\eps^{-2}\sqrt{k},\eps^{-1.5}\sqrt{n}))$ space and costs $\Ohtilde(\eps^{-3})$ time per character, which improves upon~\cite{HDstream,DBLP:journals/corr/abs-1907-04405}.

\section{Algorithmic Overview and Organization}
The central technical contribution of our work is a randomized scheme of selecting positions in a given string $S\in \Sigma^n$ so that if $f(S)\subseteq\{1,\ldots,n\}$ is the set of selected positions,
 then the following properties hold:  $|f(S)|=\Ohtilde(k)$ with high probability, the selection is preserved by rotations (the selected positions are shifted along with the underlying characters), and $|f(S)\cap f(T)|\ge k$ with high probability for every $T\in \Sigma^n$ such that $\HAM(S,T)\le k$.

Unfortunately, for integer exponents $\alpha\gg k$, such a selection of positions is infeasible for strings of the form $S=Q^\alpha$ (that we call \emph{high powers}), which are fixed points of $\shft^{n/\alpha}$. Moreover, the selection of positions is also infeasible for strings with a relatively small Hamming distance to some high power.
Hence, we define the problematic strings to be \emph{pseudo-periodic},
exclude them from the selection scheme, and deal with them~separately.

\subparagraph*{Sketches for non-pseudo-periodic strings.}
In Section~\ref{sec:nap}, we construct sketches for non-pseudo-periodic strings using a selection function $f$
satisfying the aforementioned properties.

Our $(\eps,k)$-ACS sketch stores (non-circular) approximate Hamming distance sketches of $\shft^i(S)$
for a random sample of $\Ohtilde(\sqrt{k})$ positions $i\in f(S)$.
Given the $(\eps,k)$-ACS sketches of two strings $S_1,S_2$ and a shift value $m$ such that $\HAM(S_1,\shft^{m}(S_2))\le k$, with high probability, there is a shift $i$ such that
the non-circular sketches of both $\shft^i(S_1)$ and $\shft^{i+m}(S_2)$ are available.
The decoder uses these approximate Hamming distance sketches to approximate $\HAM(\shft^i(S_1),\shft^{i+m}(S_2))=\HAM(S_1,\shft^{m}(S_2))$;
see Section~\ref{sec:approx-nap} for details.

Our $k$-ECS sketch, for each position $i\in f(S)$, stores a (non-circular) sketch of $\shft^i(S)$ capable of retrieving each mismatch with probability $\Theta(\frac {\log n}k)$, but no more than $\Oh(\log n)$ mismatches in total.
Given circular sketches of two strings $S_1,S_2$ and a shift value $m$ such that $\HAM(S_1,\shft^{m}(S_2))\le k$, with high probability,
there are at least $k$ shifts $i$ such that
the non-circular sketches of both $\shft^i(S_1)$ and $\shft^{i+m}(S_2)$ are available.
Each of these $k$ pairs of non-circular sketches yields  random mismatches
between $S_1$ and $\shft^{m}(S_2)$. Consequently, with high probability, each mismatch between $S_1$ and $\shft^{m}(S_2)$
is reported at least once, which allows for the exact computation of $\HAM(S_1,\shft^{m}(S_2))$; see Section~\ref{sec:exact-nap} for details.

\subparagraph*{Selection function.}
The selection function $f$ for non-pseudo-periodic strings is constructed in Section~\ref{sec:positions}.
Our baseline solution is to sample strings of length $\frac{n}{\gamma k}$ (for a constant $\gamma$ fixed in Section~\ref{sec:nap})
with rate $\Ohtilde(\frac{k}{n})$ and, for each sampled string $u$, to add to $f(S)$ the positions where $u$ occurs in $S$.
Unfortunately, since substrings could have much more than $\gamma k$ occurrences, the variance of $|f(S)|$ could be rather large,
and thus substrings with a large number of occurrences need to be excluded from the sample.
This workaround is feasible unless highly periodic regions cover most positions of $S$; see Section~\ref{sec:positions-np}, where the properties of $f$ are proved using concentration arguments (the Chernoff--Hoeffding bound).

In the complementary case of strings mostly covered by highly periodic regions,
we utilize the structure of these regions to deterministically select positions.
If there are many disjoint regions, it suffices to select the boundaries of the regions.
However, in general  we follow a more involved approach inspired by~\cite{BWK19,CKW20}: periodic regions are extended as long as the number of mismatches between the extended region and the period of the region is relatively small compared to the length of the extended region.
The positions of these mismatches are also added to $f(S)$.
Selection of $f$ in this case is the most technically challenging component of our construction; see Section~\ref{sec:positions-p} for~details.

\subparagraph*{Sketches for pseudo-periodic strings.}
Each pseudo-periodic string can be assigned to the nearest high power (the \emph{base})
so that two pseudo-periodic strings $S_1,S_2$  satisfy $\HAM(S_1,S_2)\le k$ only if they share the same base.
Thus, we first design a $0$-mismatch circular sketch (of size $\tilde O(1)$) to be used for comparing the bases
both in the exact and approximate variants.

Our exact $k$-mismatch circular sketch stores the mismatches between the string and its base.
Once the decoder verifies that  $S_1$ and $\shft^m(S_2)$ share the same base,
the mismatches between $S_1$ and $\shft^m(S_2)$ are reconstructed from the mismatches between each of the strings $S_1,\shft^m(S_2)$ and their common base.
The $(\eps,k)$-ACS sketch stores only the mismatches between the string and its base
at $\Ohtilde(\frac{n}{\eps \sqrt{k}})$ sampled positions (so that  $\Ohtilde(\eps^{-1}\sqrt{k})$ mismatches are stored with high probability). Once the decoder verifies that  $S_1$ and $\shft^m(S_2)$ share the same base, the mismatches between $S_1$ and $\shft^m(S_2)$ at $\Ohtilde(\frac{n}{\eps^2 k})$
jointly sampled positions are retrieved to estimate $\HAM(S_1,\shft^m(S_2))$; see~Section~\ref{sec:ap-sketches}.

\subparagraph*{Organization.}
In Section~\ref{sec:nap} and Section~\ref{sec:positions}, we describe the main novel ideas and techniques of this paper, which are used in sketches for strings that are not pseudo-periodic. In Section~\ref{sec:ap-sketches}, we provide sketches for pseudo-periodic strings, and in Section~\ref{sec:summary} we combine the sketches of Section~\ref{sec:nap}
with the sketches of Section~\ref{sec:ap-sketches} in order to prove the main theorems.
Notice that these two cases require a slight overlap so that whenever $\HAM(S_1,\shft^m(S_2))\le k$,
one of the cases accommodates \emph{both} $S_1$ and $S_2$.
In Section~\ref{sec:summary}, we also develop another $(\eps,k)$-ACS sketch, tailored to approximating large distances.
This simple construction improves the size of $(\eps,k)$-ACS sketches (for $k\ge \eps n$) from  $\Ohtilde(\eps^{-2}\sqrt{k})$ to $\Ohtilde(\eps^{-1.5}\sqrt{n})$.
Finally, in Section~\ref{sec:shift-decoding}, we describe efficient decoding algorithms for retrieving the shift distance from the encodings developed for the circular $k$-mismatch sketches.

\section{Preliminaries}\label{sec:preliminaries}
For integers $\ell \le r$, we denote  $[\ell\dd r]=\{\ell,\ell+1,\ldots, r\}$.
Moreover, $[n] = [1\dd n]$.

A string $S$ of length $|S|=n$ is a sequence of characters $S[1]S[2]\cdots S[n]$ over an alphabet $\Sigma$;
in this work, we assume that $\Sigma=[\sigma]$. The set of all length-$n$ strings over $\Sigma$ is denoted by $\Sigma^n$.
A string $T$ is a \emph{substring} of a string $S\in \Sigma^n$ if $T=S[i]S[i+1]\cdots S[j]$ for $1\le i\le j \le n$.
In this case, we denote the occurrence of $T$ at position $i$ by $S[i\dd j]$.
Such an occurrence is a \emph{fragment} of $S$.
A fragment $S[i\dd j]$ is a \emph{prefix} of $S$ if $i = 1$ and a \emph{suffix} of $S$ if $j=n$.


\subparagraph*{Hamming distance.}
The \emph{Hamming distance} $\HAM(S,T)$ of two strings $S,T\in \Sigma^n$ is defined as the number of positions $i\in [n]$ such that $S[i]\ne T[i]$. We denote $\MP(S,T)=\{i\in [n] \mid S[i]\ne T[i]\}$ to be the set of \emph{mismatch positions} and $\MI(S,T)=\{(i,S[i],T[i]) \mid i\in [n], S[i]\ne T[i]\}$ to be the underlying \emph{mismatch information}. Note that $\HAM(S,T)=|\MP(S,T)|=|\MI(S,T)|$.

For a subset $A\sub [n]$, we denote $\MI_A(S,T)=\{(i,a,b)\in \MI(S,T) \mid i\in A\}$
and $\HAM_A(S,T)=|\MI_A(S,T)|$.
The following result, based on the Chernoff bound, shows that $\HAM_A(S,T)$ for random $A$ yields an approximation of $\HAM(S,T)$.
\begin{lemma}\label{lem:chernoff}
	Let $A$ be a random subset of $[n]$ with elements chosen independently at rate $p$.
	For $0<\eps < 1$, we have $\Pr[\HAM_A(S,T)\in (1\pm \eps)p\HAM(S,T)] \ge 1- 2\exp\left(-\tfrac{p\HAM(S,T)\eps^2}3\right)$.
\end{lemma}
\begin{proof}
	For each index $i\in[n]$, let $x_i$ be an indicator variable such that $x_i=1$ if $i\in\MI_A(S,T)$ and $x_i=0$ otherwise.
	Note that $\HAM_A(S,T)=|\MI_A(S,T)|=\sum_{i=1}^n x_i$ and that $x_i$ are independent variables.
	For every $i\in\MI(S,T)$, we have $\Pr[x_i=1]=p$ and, for every $i\notin\MI(S,T)$, we have $\Pr[x_i=1]=0$.
	Thus, $\mathbb E[\sum_{i=1}^n x_i]=\mathbb E[\sum_{i\in \MI_A(S,T)} x_i]=p\HAM_A(S,T)$.
	Hence, by the Chernoff bound (see, e.g.,~\cite{Doerr})
	\[\Pr\left[\left|\sum_{i=1}^n x_i-p\HAM_A(S,T)\right|>\eps p\HAM_A(S,T)\right]\le 2\exp\left(-\frac{p\HAM_A(S,T)\eps^2}3 \right).\]
	Thus,  $\Pr[\HAM_A(S,T)\in (1\pm \eps)p\HAM(S,T)] \ge 1- 2\exp(-
	\tfrac{p\HAM(S,T)\eps^2}3)$.
\end{proof}
The triangle inequality yields $\HAM(S,U)\le \HAM(S,T)+\HAM(T,U)$ for $S,T,U\in \Sigma^n$.
The underlying phenomenon also allows retrieving $\MI(S,U)$ from $\MI(S,T)$ and $\MI(T,U)$. The following fact is proved in the following.
\begin{fact}\label{fct:retrieve}
	For every $S,T,U\in \Sigma^n$ and every $A\sub [n]$, the mismatch information $\MI_A(S,U)$ can be retrieved from $\MI_A(S,T)$ and $\MI_A(T,U)$ in time $\Ohtilde(\HAM_A(S,T)+\HAM_A(T,U))$.
\end{fact}
\begin{proof}
	For each $i\in A$, we have one of the following four cases:
	\begin{itemize}
		\item if $i\notin \MP(S,T)$ and $i\notin \MP(T,U)$, then $S[i]=T[i]=U[i]$, so $i\notin \MP(S,U)$,
		\item if $(i,a,b)\in \MI(S,T)$ and $i\notin \MP(T,U)$, then $S[i]=a\ne b = T[i]=U[i]$, so $(i,a,b)\in \MI(S,U)$,
		\item if $i\notin \MP(S,T)$ and $(i,b,c)\in \MP(T,U)$, then $S[i]=T[i]=b \ne c = U[i]$, so $(i,b,c)\in \MI(S,U)$,
		\item if $(i,a,b)\in \MI(S,T)$ and $(i,b,c)\in \MP(T,U)$, then $S[i]=a \ne b = T[i] = b \ne c = U[i]$, so $(i,a,c)\in \MI(S,U)$ (if $a\ne c$) or $i\notin \MP(S,U)$ (if $a = c$).\qedhere
	\end{itemize}
\end{proof}

\subparagraph*{Periods.}
An integer $p$ is a \emph{period} of $S\in\Sigma^*$ if and only if $S[i] = S[i+p]$ for all $1\le i \le |S|-p$.
The shortest period of $S$ is denoted $\per(S)$.
If $\per(S)\leq \frac12 |S|$, we say that $S$ is \textit{periodic}.

\subparagraph*{Rotations.}
For a string $S=S[1]S[2]\cdots S[n]$, let $\shft(S)=S[2]\cdots S[n]S[1]$.
For $i\in \mathbb{Z}$, we denote $i\cmod n=((i-1)\modulo n)+1$
so that, for $i\in [n]$, the value $(i-1)\cmod n$ is the position of $S[i]$ in $\shft(S)$.%
\footnote{We introduce the $\cmod$ operator because positions in strings are indexed from $1$ rather than from $0$.}
Moreover, for $M\subseteq\mathbb Z$, we denote $M\cmod n=\{i\cmod n\mid i\in M\}$.
For $P\subseteq [n]$, let $\rot(P)=\{(i-1)\cmod n\mid i\in P\}$ be the rotated set $P$.

The \emph{primitive root} of a string $S$ is the shortest string $Q$ such that
$S=Q^\alpha$ for an integer $\alpha\ge 1$. The length of the primitive root is denoted by $\prr(S)$.
Notice that $\per(S)\le \prr(S)$.
Moreover, for every $m,m'\in \mathbb{Z}$, we have that $\prr(\shft^m(S)) = \prr(S)$, and $\shft^{m}(S)=\shft^{m'}(S)$ if and only if $\prr(S) \mid (m-m')$.

\section{Sketches for Non-pseudo-periodic Strings}\label{sec:nap}
We say that a string $S\in \Sigma^n$ is \emph{$(\alpha, \beta)$-pseudo-periodic} if there exists a string $S'\in \Sigma^n$, called an \emph{$(\alpha,\beta)$-base} of $S$, such that $\prr(S')\le \frac{n}{\alpha}$ and $\HAM(S,S')\le \beta$.
\begin{observation}\label{obs:rootrot}
	If $S$ is $(\alpha,\beta)$-pseudo-periodic with an $(\alpha,\beta)$-base $S'$, then every rotation $\shft^m(S)$ with $m\in \mathbb{Z}$ is also $(\alpha,\beta)$-pseudo-periodic and $\shft^m(S')$ is an $(\alpha,\beta)$-base of $\shft^m(S)$.
\end{observation}

Let $\mathcal{H}_{n,k}$ be the set of strings in $\Sigma^n$ that are $(3\gamma k,\gamma k)$-pseudo-periodic, where $\gamma$ is the smallest constant such that $\gamma \ge 14$ and $\tfrac{n}{3\gamma k}$ is an integer.
In this section, we present two circular sketches for strings in $\Sigma^n\sm \H_{n,k}$: an $(\eps,k)$-ACS sketch and a $k$-ECS sketch.
Both sketches rely on the following result, proved in Section~\ref{sec:positions}.
\begin{theorem}\label{thm:points-set}
	For every two integers $1\le k \le n$, there exists a randomized function
	$f:\Sigma^n\setminus\mathcal{H}_{n,k} \rightarrow 2^{[n]}$
	such that the following holds for every $S_1,S_2\in \Sigma^n\setminus\mathcal{H}_{n,k}$:
	\begin{enumerate}
		\item\label{prop:bounded-size} $|f(S_1)|=\Ohtilde(k)$  with high probability,
		\item\label{prop:rotation} $f(\shft(S_1))=\rot(f(S_1))$,
		\item\label{prop:large-intersection} if $\HAM(S_1,S_2)\le k$, then $|f(S_1)\cap f(S_2)|\ge k$ with high probability.
	\end{enumerate}
\end{theorem}

\subsection{An $(\eps,k)$-ACS Sketch}\label{sec:approx-nap}
We start with briefly presenting a useful technical tool, that is, the non-circular version of the approximate sketch. We remark that many variants of this sketch exist, with equivalent space complexity.
A short proof is given for the sake of completeness.

\begin{theorem}[$(1\pm\eps)$-approximate sketches, folklore]\label{thm:approx-sketch-non-circular}
	There exists a $(1\pm\eps)$-approximate sketch $\sk_{\varepsilon}$ such that, given $\sk_{\varepsilon}(S_1)$ and $\sk_{\varepsilon}(S_2)$ for two strings $S_1,S_2 \in \Sigma^n$, one can decode $\HAM(S_1,S_2)$ with a $(1 \pm \eps)$-multiplicative error. The sketches use $\Ohtilde(\eps^{-2})$ space, the decoding algorithm is correct with high probability and costs $\Ohtilde(\eps^{-2})$ time.
\end{theorem}
\begin{proof}
	Consider $\mu : \Sigma \to \{0,1\}^{\sigma}$ defined as $\mu(c) = 0^{c-1}10^{\sigma-c}$. For every words $u,v$, we have $\HAM( \mu(u), \mu(v) ) = 2 \cdot \HAM(u,v)$.
	We then use AMS sketches~\cite{ALON1999137} on $\mu(u)$ and $\mu(v)$ which allow for decoding of $\ell_2$ distance $\|\mu(u)-\mu(v)\|_2$.
	This is enough since, for binary words, the $\ell_2^2$ distance coincides with the Hamming distance.
	We then note that the AMS sketches of $\mu(u)$ and $\mu(v)$ can be computed without explicitly constructing $\mu(u)$ or $\mu(v)$.
\end{proof}

Next, we describe our sketching scheme and prove that, together with an appropriate decoding algorithm, it forms an $(\eps,k)$-ACS sketch for $\Sigma^n\setminus\H_{n,k}$.
\begin{construction}\label{def:sketchesapproxcirc}
	The encoding function $\circ_{\eps,k}:\Sigma^n\sm \H_{n,k} \to \{0,1\}^*$ is defined as follows:
	\begin{enumerate}
		\item Let  $f:\Sigma^n\sm \H_{n,k}\to 2^{[n]}$ be the selection function of Theorem~\ref{thm:points-set}.
		\item Let $\sk_{\eps} : \Sigma^n \to \{0,1\}^*$ be the sketch of Theorem~\ref{thm:approx-sketch-non-circular}.
		\item Let $A,B\sub [n]$ be two subsets\footnote{%
		 The sketch would remain valid with one subset only. However, introducing the second subset simplifies the arguments
		 and makes the construction more similar to the counterpart for pseudo-periodic strings.} with elements sampled independently with rate $p = 2 \sqrt{\frac{\ln n}{k}}$.
		\item For $S\in\Sigma^n\sm \H_{n,k}$, the encoding $\circ_{\eps,k}(S)$ stores $(i,\sk_{\varepsilon}(\shft^i(S)))$ for $i\in f(S)\cap (A\cup B)$.
	\end{enumerate}
\end{construction}
\begin{proposition}\label{thm:approx-nap-circ-sketch}
	There exists a decoding function which, together with the encoding $\circ_{\varepsilon,k}$  of Definition~\ref{def:sketchesapproxcirc}, forms an $(\eps,k)$-ACS sketch of $\Sigma^n \sm \H_{n,k}$.	The size of this sketch is $\Ohtilde(\eps^{-2}\sqrt{k})$, and the decoding algorithm costs $\Ohtilde(\sqrt{k}+\eps^{-2})$ time with high probability.
\end{proposition}
\begin{proof}
	Our decoding procedure iterates over $i \in f(S_1)\cap A$. If $i' := (i+m) \cmod n \in f(S_2)\cap B$, the procedure retrieves the sketches $\sk_{\varepsilon}(\shft^{i}(S_1))$ and $\sk_{\varepsilon}(\shft^{i'}(S_2))$ and recovers a $(1+\eps)$-approximation of $\HAM(\shft^i(S_1),\shft^{i'}(S_2)) = \HAM(S_1,\shft^{m}(S_2))$. Otherwise, $\infty$ is returned.
	
	We now reason that if $\HAM(S_1,\shft^m(S_2)) \le k$, then, with high probability, $i'\in  f(S_2)\cap B$ for some $i\in f(S_1)\cap A$. By Theorem~\ref{thm:points-set}, $|f(S_1) \cap f(\shft^m(S_2))| \ge k$. Thus, for any $i \in f(S_1) \cap f(\shft^m(S_2))$, we have that $i \in f(S_1) \cap A$ with probability $p$. Similarly, $i' \in f(S_2) \cap B$ with probability $p$. Since $A$ and $B$ are independent, we have a success probability $p^2$ for each $i$ independently. The probability of at least one success is at least $1 - (1 - p^2)^k \ge 1 - n^{-4}$.
	
	The decoding time is given by the time needed to compute the intersection of $f(S_1)\cap A$ and $\rot^m(f(S_2)\cap B)$,  which is $\Ohtilde(\sqrt{k})$ with high probability, and $\Ohtilde(\eps^{-2})$ time to decode the distance from a single pair of indices $i,i'$, provided that the intersection is not empty.
\end{proof}

\subsection{An $k$-ECS Sketch}\label{sec:exact-nap}
We begin with the following corollary of~\cite[Theorem 5.1]{PL07}. The original statement in~\cite{PL07} is given for $A = [n]$ only, but it can be generalized in a straightforward manner, e.g., by replacing all characters at positions in $[n]\sm A$ with a fixed character.
\begin{theorem}[based on~{\cite[Theorem 5.1]{PL07}}]\label{thm:exact-non-circular}
	For every $k \le n$ and $A \subseteq [n]$, there is a sketch $\sk_{k,A}$ of size $\Ohtilde(k)$ such that, given $\sk_{k,A}(S_1)$ and $\sk_{k,A}(S_2)$ for two strings $S_1,S_2 \in \Sigma^n$:
	\begin{itemize}
		\item if $\HAM_A(S_1,S_2) \le k$, then the decoding function returns $\MI_A(S_1,S_2)$;
		\item otherwise, if $\HAM_A(S_1,S_2) > k$, the decoding function reports that this is the case.
	\end{itemize}
	The decoding algorithm is correct with high probability and costs $\Ohtilde(k)$ time.
\end{theorem}

\begin{construction}\label{def:sketchesexactcirc}
	The encoding function $\circ_{k} : \Sigma^n\setminus\H_{n,k} \to \{0,1\}^*$ is defined as follows:
	\begin{enumerate}
		\item Let  $f:\Sigma^n\sm \H_{n,k}\to 2^{[n]}$ be the selection function of Theorem~\ref{thm:points-set}.
		\item Let $A\sub [n]$ be a subset with elements sampled independently with rate $p := \frac{9 \ln n}{k}$.
		\item Denote $t = \lceil18 \ln n\rceil$, and let $\sk_{t,A} : \Sigma^n \to \{0,1\}^*$ be the sketch of Theorem~\ref{thm:exact-non-circular}.
		\item For $S\in\Sigma^n\sm \H_{n,k}$, the encoding $\circ_{k}(S)$ stores the pairs $(i, \sk_{t,A}(\shft^i(S)))$ for $i\in f(S)$.
	\end{enumerate}
\end{construction}

\begin{proposition}\label{thm:exact-nap-sketch}
	There exists a decoding function which, together with the encoding $\circ_{k}$  of Definition\ref{def:sketchesexactcirc}, forms a $k$-ECS sketch of $\Sigma^n \sm \H_{n,k}$.
	The size of this sketch is $\Ohtilde(k)$, and the decoding algorithm costs $\Ohtilde(k)$ time with high probability.
\end{proposition}
\begin{proof}
	The decoding procedure iterates over $i\in f(S_1)\cap \rot^m(f(S_2))$,
	If the number of such positions is less than $k$, then $\infty$ is returned.
	Otherwise, for each $i\in f(S_1)\cap \rot^m(f(S_2))$, we have $i' := (i+m)\cmod n \in f(S_2)$,
	and the algorithm runs a decoding procedure for $\sk_{t,A}(\shft^i(S_1))$ and $\sk_{t,A}(\shft^{i'}(S_2))$.
	If any such decoding fails, then $\infty$ is returned.
	Otherwise, for each mismatch position $j$ found, say with $\shft^i(S_1)[j] \ne \shft^{i'}(S_2)[j]$, the algorithm adds $(i + j) \cmod n$ to a set $M$, initialized as the empty set. Finally, the size $|M|$ is returned.
	
	The decoding procedure costs $\Ohtilde(k)$ time, which is needed both to find all the aligned pairs $i\in f(S),i'\in f(S_2)$ by computing the intersection $f(S_1)\cap \rot^m(f(S_2))$ and to retrieve and gather the mismatches obtained from the aligned pairs (in $\Ohtilde(t)=\Ohtilde(1)$ time per pair).

	\subparagraph*{Correctness.}
	Recall that $\MP(S_1,\shft^m(S_2))$ is the set of mismatch positions between $S_1$ and $\shft^m(S_2)$.
	First, notice that each $j\in M$ is a mismatch position between $S_1$ and $\shft^m(S_2)$, since $\shft^i(S_1)[j] \ne \shft^{i'}(S_2)[j]$ is equivalent to $S_1[(i+j)\cmod n] \ne S_2[(i+j+m)\cmod n]$.
	Hence, $M\subseteq \MP(S_1,\shft^m(S_2))$ and $|M|\le \HAM(S_1,\shft^m(S_2))$
	
	Now, we prove that, with high probability, if $\HAM(S_1,\shft^m(S_2))\le k$, then the algorithm reports $\HAM(S_1,\shft^m(S_2))$, and if $\HAM(S_1,\shft^m(S_2))>k$, then the algorithm reports a value larger than $k$.
	In the case where $\HAM(S_1,\shft^m(S_2))\le k$, we have that $|f(S_1)\cap \rot^m(f(S_2))|\ge k$ with high probability due to Theorem~\ref{thm:points-set}. Moreover, for every $i\in f(S_1)\cap \rot^m(f(S_2))$, the expected number of positions in $A\cap\MP(S_1,\shft^m(S_2))$ is $\HAM(S_1,\shft^m(S_2))\cdot p\le k\cdot\frac{9\ln n}{k}=9\ln n$.
	Hence, by a Chernoff bound $\Pr[|A\cap\MP(S_1,\shft^m(S_2))|>18\ln n]\le \exp(-\frac{9 \ln n}3) = n^{-3}$.
	Thus, when $\HAM(S_1,\shft^m(S_2))\le k$, decoding $\sk_{t,A}(\shft^i(S_1))$ and $\sk_{t,A}(\shft^{i'}(S_2))$
	succeeds for all $i\in f(S_1)\cap \rot^m(f(S_2))$ with high probability.
	
	Conditioned on the event that $|f(S_1)\cap \rot^m(f(S_2))|\ge k$ and the decoding algorithm of $\sk_{t,A}$ is successful, we now prove that $|M|=\HAM(S_1,\shft^m(S_2))$.
	For each mismatch $j\in \MP(S_1,\shft^m(S_2))$, there is an independent trial associated with each $i\in f(S_1) \cap \rot^m(f(S_2))$, which is whether $((j-i) \cmod n) \in A$ or not. The trial is successful with probability $p$.
	The probability that at least one of those trials succeeds is at least $1 - (1-p)^k \ge 1 - n^{-9}$.
	Applying the union bound over all $j\in\MP(S_1,\shft^m(S_2))$, we conclude that
	 $M=\MP(S_1,\shft^m(S_2))$ and $|M|=\HAM(S_1,\shft^m(S_2))$ with high probability.
	
	If $\HAM(S_1,\shft^m(S_2))>k$, then the decoding algorithm may return $\infty$ because of $|f(S_1)\cap \rot^m(f(S_2))|< k$ or due to a decoding failure. If neither of these events happen, the algorithm returns $|M|$, which is equal to $\HAM(S_1,\shft^m(S_2))$ with high probability (as proved above). Thus, in both cases, a value larger than $k$ is reported.
\end{proof}

\section{Construction of the Selection Function}\label{sec:positions}
For $S\in\Sigma^n$, let $S^*=S\cdot S\cdot S\cdots$ be the infinite string which is the infinite concatenation of $S$ to itself (for any $i\in\mathbb N$, we have $S^*[i]=S[i\cmod n]$).
Let $\ell=\frac n{3\gamma k}$ (recall it is an integer).
A position $i\in[n]$ is called \emph{cubic} if $u_i=S^*[i\dd i+3\ell-1]$ has $\per(u_i)\le \frac {|u_i|}{3}=\ell$, i.e., if the cyclic fragment of length $3\ell$ starting at position $i$ consists of at least three repetitions of the same factor. Otherwise, position $i$ is called \emph{non-cubic}.
We denote the set of cubic positions in a string $S$ as $\P(S)$, and the set of non-cubic positions as $\NP(S)$. Notice that $\P(S)\cup\NP(S)=[n]$ and $\P(S)\cap\NP(S)=\emptyset$.

We present two selection techniques, resulting in functions $\fnp$ and $\fp$,
designed for strings with many non-cubic positions and for strings with many cubic positions, respectively.
Both functions satisfy the first two properties of Theorem~\ref{thm:points-set} for any string $S_1\in\Sigma^n\setminus\H_{n,k}$.
The functions $\fnp$ and $\fp$ have the third property of Theorem~\ref{thm:points-set} if $|\NP(S_1)|\ge\frac n2$
and if $|\P(S_1)|\ge \frac n2$, respectively.
Thus, the function $f$ defined through $f(S)=\fnp(S)\cup\fp(S)$ satisfies Theorem~\ref{thm:points-set}.

\subsection{Selecting Positions for Strings with Many Non-cubic Positions}\label{sec:positions-np}
Throughout this subsection, let  $h:\Sigma^{3\ell}\rightarrow\{0,1\}$ be a hash function assigning values independently to each $u\in\Sigma^{3\ell}$ such that $\Pr[h(u)=1]=\frac {4 k\ln n}n$.
For clarity, we omit the explicit dependence on $h$ in our notation.
For $S\in\Sigma^n$, define $\fnp(S)=\left\{i\in\NP(S)\mid h(u_i)=1\right\}$.

Our proofs rely on the following multiplicative Chernoff--Hoeffding bound:
\begin{proposition}[Corollary of {\cite[Theorems 1.10.1 and 1.10.5]{Doerr}}]\label{prp:chernoff}
	Let $X_1,\ldots,X_n$ be independent random variables taking values in $[0,M]$,
	 let $X = \sum_{i=1}^n X_i$, and let $\mu \ge 0$.
	\begin{enumerate}
	\item\label{it:ub} If $\mu \ge \Exp[X]$, then, for every $\delta>0$, we have
	$\Pr[X \ge (1+\delta)\mu]\le \exp(-\frac{\min(\delta,\delta^2)\mu}{3M})$.
	\item\label{it:lb} If $\mu \le \Exp[X]$, then, for every $0<\delta<1$, we have
	$\Pr[X \le (1-\delta)\mu]\le \exp(-\frac{\delta^2\mu}{2M})$.
	\end{enumerate}
\end{proposition}

We first prove that $\fnp$ satisfies the first property of Theorem~\ref{thm:points-set}.
\begin{lemma}\label{lem:np_ub}
For every $S\in\Sigma^n$, we have $\Pr\left[|\fnp(S)|< 8 k\ln n\right]\ge 1-n^{-\Omega(1)}$.
\end{lemma}
\begin{proof}
For each $u\in \Sigma^{3\ell}$, we introduce a random variable $X_u = |\{i \in \fnp(S) \mid u_i = u\}|$;
notice that $X_u$ depends only on $h(u)$, so the variables $X_u$ are independent.
In order to apply Property~\ref{prp:chernoff} for $|\fnp(S)|=\sum_{u \in \Sigma^{3\ell}} X_u$,
we prove that each $X_u$ is bounded.

First, note that if $\per(u)\le \ell$ or $h(u)=0$, then $X_u = 0$.
Otherwise, as $u_{i}=u=u_{i'}$ for $i < i' \le i+3\ell$ implies
$i'-i\ge \per(u) > \ell$, we conclude that $X_u=|\{i\in [n] \mid u_i = u\}|\le \frac{n}{\ell}= 3\gamma k$. 
Now,
$\Exp[|\fnp(S)|]=\sum_{i\in \NP(S)}\Pr[h(u_i)=1]=|\NP(S)|\cdot \tfrac{4 k \ln n}{n} \le 4 k \ln n$,
so, by Property~\ref{prp:chernoff}\eqref{it:ub} with $\delta=1$, we have
$\Pr[|\fnp(S)| \ge 8 k \ln n] \le \exp(-\tfrac{4 k \ln n}{3\cdot 3 \gamma k})
= n^{-4/(9\gamma)}=n^{-\Omega(1)}$.
\end{proof}

The following lemma states that $\fnp$ satisfies Property~\ref{prop:rotation} of Theorem~\ref{thm:points-set}.

\begin{lemma}\label{lem:np-rotation}
For every $S\in \Sigma^n$, we have $\fnp(\shft(S))=\rot(\fnp(S))$.
\end{lemma}
\begin{proof}
	Let $i\in \fnp(\shft(S))$ and let $u=(\shft(S))^*[i\dd i+\ell-1]=S^*[i+1\dd i+\ell]$. Since $i\in\fnp(\shft(S))$,
	we have that $\per(u)>\frac \ell3$ and $h(u)=1$.
	Therefore, $(i+1)\cmod n\in \fnp(S)$, which means that  $i\cmod n = i \in \rot(\fnp(S))$.
	Hence, $\fnp(\shft(S))\subseteq \rot(\fnp(S))$.
	Symmetrically, $\rot(\fnp(S)) \subseteq \fnp(\shft(S))$. Thus, $\fnp(\shft(S))=\rot(\fnp(S))$.
\end{proof}

Finally, the following lemma states that $\fnp$ satisfies Property~\ref{prop:large-intersection} of Theorem~\ref{thm:points-set}.

\begin{lemma}\label{lem:fnp-large-intersection}
	Suppose that $S_1,S_2\in \Sigma^n$ satisfy $\HAM(S_1,S_2)\le k$.
	If  $|\NP(S_1)|\ge \frac 12n$, then $\Pr[|\fnp(S_1)\cap \fnp(S_2)|\ge k]\ge 1-n^{-\Omega(1)}$.
\end{lemma}
\begin{proof}
For each $i\in [n]$, let $u_i=S_1^*[i\dd i+3\ell-1]$ and $v_i=S_2^*[i\dd i+3\ell-1]$,
and let $\Lambda=\{i\in \NP(S_i)\mid u_i=v_i\}$.
Notice that, for $i\in[n]$, we have
$u_i\ne v_i$ if and only if $\MP(S_1,S_2)\cap ([i\dd i+3\ell-1]\cmod n)\ne\emptyset$.
Hence, the number of indices $i\in[n]$ with $u_i\ne v_i$ is at most $|\MP(S_1,S_2)|\cdot 3\ell\le k\cdot\frac n{\gamma k}\le \frac n{\gamma}$.
Since $|\NP(S_1)|\ge \frac 12n$, then $|\Lambda|\ge\frac 12n-\frac 1{\gamma}n> \frac13 n$ due to $\gamma \ge 6$.
Thus, $\Exp[|\fnp(S_1)\cap \fnp(S_2)|]\ge|\Lambda|\cdot \tfrac{4k\ln n}{n}\ge \tfrac43 k\ln n$. The rest of the proof follows from Property~\ref{prp:chernoff}\eqref{it:lb} similarly as Property~\ref{prp:chernoff}\eqref{it:ub} is applied in the proof of Lemma~\ref{lem:np_ub}.
\end{proof}

\subsection{Selecting Positions for Strings with Many Cubic Positions}\label{sec:positions-p}
Recall that our goal is to design a rotation-invariant mechanism for selecting $\Ohtilde(k)$ indices so that, given two fairly similar strings, at least $k$ common indices are selected in both strings.
In the selection procedure described in Section~\ref{sec:positions-np}, the decision whether or not to include position $i$ was based on whether or not $S^*[i\dd i+3\ell-1]\in \S$ for a certain family $\S\sub \Sigma^{3\ell}$. Then, we argued that $S_1^*[i\dd i+3\ell-1]=S_2^*[i\dd i+3\ell-1]\in \S$ for at least $k$ positions $i\in [n]$.

Unfortunately, this strategy might be infeasible if $\P(S)$ is large, that is, when there is a large number of cubic positions in $S$.
For example, it could be the case that $S_1^*[i\dd i+3\ell-1]\ne S_2^*[i\dd i+3\ell-1]$ holds for $3\ell k=\frac{n}{\gamma}$
positions $i\in[n]$, and $S_1^*[i\dd i+3\ell-1]= S_2^*[i\dd i+3\ell-1]=\texttt{a}^{3\ell}$ for the remaining $n-\frac{n}{\gamma}$ positions $i\in [n]$.
This may happen even if $\HAM(S_1,\texttt{a}^n)=\Omega(\frac{n}{\gamma})$, i.e., for strings far from being $(3\gamma k, \gamma k)$-pseudo-periodic.

We begin with some intuition for the construction of the function $\fp$.
First, suppose that, for each position $i\in \P(S)$, we include in $\fp(S)$ the smallest $j>i$ such that $\per(S^*[i\dd j])>\per(S^*[i\dd i+3\ell-1])$. In other words, $\fp(S)$ contains the positions following each maximal cyclic fragment of length at least $3\ell$ and period at most $\ell$.
Notice that this construction satisfies Property~\ref{prop:rotation} of Theorem~\ref{thm:points-set}. Moreover, since each position may belong to at most two such maximal repetitions,
the number of positions selected is at most $\frac{2n}{3\ell}=2\gamma k$ (so that Property~\ref{prop:bounded-size} of Theorem~\ref{thm:points-set} is satisfied), and a substitution of a single character in $S$ may remove at most two positions from $\fp(S)$.
However, if the cubic positions are clustered in few blocks, then this mechanism is not enough to guarantee that  Property~\ref{prop:large-intersection} of Theorem~\ref{thm:points-set} is satisfied, i.e., that $|\fp(S_1)\cap \fp(S_2)|\ge k$ when $\HAM(S_1,S_2)\le k$. Hence, instead of selecting just one position $j$ for each $i\in \P(S)$, several positions are selected using a process inspired by~\cite{BWK19} with subsequent improvements in~\cite{CKW20}:
The fragment $S^*[i\dd i+3\ell-1]$ is maximally extended
to $S^*[i\dd i+\tau_i-1]$ so that the period of $S^*[i\dd i+\tau_i-1]$ drops to~$\per(S^*[i\dd i+3\ell-1])$ after $\Theta(\frac{k}{n}\tau_i)$ substitutions, and the underlying mismatching positions are added to $\fp(S)$.

\subsubsection{Definition of $\fp$}
For any $i\in\P(S^*)$, let $u_i=S^*[i\dd i+3\ell-1]$, let  $\rho_i=\per(u_i)$,
and let $\mu_{S,i}=S^*[i\dd i+\rho_i-1]$, which is the string period of $u_i$.
To avoid clutter in the presentation, we use $\mu_i=\mu_{S,i}$ when $S$ is clear from context.
Notice that, for $\tau\ge 2\rho_i$, the string $\mu_i^*[1\dd \tau]$ is the (unique) string of length $\tau$ with string period $\mu_i$.

We are now ready to formally define the concept of extending (to the right) a cubic fragment starting at position $i$ for as long as the ratio between the length of the extended fragment and  the Hamming distance between the extended fragment and the appropriate prefix of $\mu_i^*$ is large enough.
The length of such a (maximal) extended fragment is defined as \[\tau_{S,i}=\min \left\{\tau\mid \tau<\tfrac { n}{\gamma k}\HAM\left(S^*[i\dd i+\tau-1],\mu_i^*[1\dd \tau]\right)\right\}.\]
The following lemma shows that $\tau_{S,i}$ is well-defined, i.e., that the minimum in the definition of $\tau_{S,i}$ is taken over a non-empty set. The bound $\tau_{S,i}\le 2n$ is also useful later on. 
\begin{lemma}\label{lem:tau-le-2n}
For every $S\in \Sigma^n\setminus\mathcal{H}_{n,k}$ and $i\in\P(S)$, we have $\tau_{S,i}\le 2n$.
\end{lemma}
\begin{proof}
	Let $i\in\P(S)$ and assume by contradiction that $\tau_{S,i}> 2n$.
	This yields \[2n \ge \tfrac{n}{\gamma k} \HAM\left(S^*[i\dd i+2n-1],\mu_i^*[1\dd 2n]\right).\]
	Moreover, $S^*[i\dd i+n-1]=S^*[i+n\dd\allowbreak i+2n-1]$, and so, by the triangle inequality,
	\begin{align*}
	2\gamma k&\ge  \HAM\left(S^*[i\dd i+2n-1],\mu_i^*[1\dd 2n]\right)\\&=\HAM\left(S^*[i\dd i+n-1],\mu_i^*[1\dd n]\right)+\HAM\left(S^*[i+n\dd i+2n-1],\mu_i^*[n+1\dd 2n]\right)\\&=\HAM\left(S^*[i\dd i+n-1],\mu_i^*[1\dd n]\right)+\HAM\left(S^*[i\dd i+n-1],\mu_i^*[n+1\dd 2n]\right)\\ &\ge \HAM\left(\mu_i^*[1\dd n],\mu_i^*[n+1\dd 2n]\right).
	\end{align*}
	Notice that  for any strings $x,y,z$ (with $|x|=|y|$) and any integer $m$, we have $\HAM(x,y)=\frac{1}{m}\HAM(x^m,y^m)$ and $\HAM(x,y)\le \HAM(xz,yz)$.
	Thus, due to $|\mu_i|=\rho_i \le \ell \le \frac{n}{3\gamma k}$, we have
	\begin{multline*}\HAM\left(\mu_i,\mu_i^*[n+1\dd n+\rho_i]\right)=\tfrac{1}{3\gamma k}\HAM\left(\mu_i^*[1\dd 3\gamma k \rho_i],\mu_i^*[n+1\dd n+3\gamma k \rho_i ]\right)\\\le \tfrac{1}{3\gamma k}\HAM\left(\mu_i^*[1\dd n],\mu_i^*[n+1\dd 2n]\right)\le \tfrac{2\gamma k}{3\gamma k} < 1.\end{multline*}
	Consequently, $\mu_i=\mu_i^*[n+1\dd n+\rho_i]=\shft^n(\mu_i)$, which implies $\rho_i \mid n$ by primitivity of $\mu_i$ (recall that $\mu_i = \shft^m(\mu_i)$ only for $\rho_i \mid m$).
	Since
	$\tau_{S,i}>n$, we have  $n \ge \frac{n}{\gamma k} \HAM(S^*[i\dd i+n-1],\allowbreak\mu_i^*[1\dd n])$, that is
	$\gamma k\ge  \HAM\big(S^*[i\dd i+n-1],\mu_i^*[1\dd n]\big)=\HAM\big(S^*[i\dd i+n-1], \mu_i^{n/\rho_i}\big)$. Hence, $S^*[i\dd i+n-1]\in \H_{n,k}$ so, by Observation~\ref{obs:rootrot}, $S\in \H_{n,k}$.
\end{proof}

Let $R_{S,i}=[i\dd i+\tau_i-1]$ be the positions in the extended fragment, and let $M_{S,i}=\{j\in R_{S,i}\mid S[j]\ne \mu_i^*[j-i+1]\}$ be the set of positions in $R_{S,i}$ corresponding to mismatches between $S^*[i\dd i+\tau_i-1]$ and $\mu_i^*[1\dd \tau_i]$. To avoid clutter in the presentation, we use $\tau_i=\tau_{S,i}$, $R_i=R_{S,i}$, and $M_i=M_{S,i}$ when $S$ is clear from context.
Define \[\fp(S)=\bigcup_{i\in \P(S)} (M_{i}\cmod n) =\{p\cmod n\mid p\in M_{i}, i\in \P(S)\}.\]

\subsubsection{Properties of $\fp$}
\subparagraph*{Property~\ref{prop:bounded-size} of Theorem~\ref{thm:points-set}.}
Our strategy for proving an upper bound on the size of $\fp(S)$
is to associate each $i\in \P(S)$ with a carefully defined set $A_i \sub R_i$.
We then select a subset $\Gamma\sub \P(S)$ so that the sets $A_i$ for $i\in \Gamma$ are disjoint subsets of $[1\dd 3n]$
and $\bigcup_{i\in \Gamma}M_i = \bigcup_{i\in \P(S)}M_i$.
Finally, we show that $|M_i| = \Oh(\frac{\gamma k}{n}|A_i|)$ for each $i\in \P(S)$,
and so $|\bigcup_{i\in \P(S)}M_i|=| \bigcup_{i\in \Gamma}M_i| = \Oh(\sum_{i\in \Gamma}\frac{\gamma k}{n}|A_i|)=\Oh(\gamma k)$.

For each $R_i$, consider the set of indices $j\in R_i$ such that $[j,j+2\ell)\cap M_i = \emptyset $.
Formally, let $A_i=\{j\in R_i \mid [j,j+2\ell)\subseteq R_i\setminus M_i\}$.
The following lemma lets us define $\fp(S)$ as the union of $M_i\cmod n$ for a restricted set of values of $i$, with the property of having disjoint sets $A_i$.

\begin{lemma}\label{lem:intersect-A-contain-M}
	Let $i,i'\in \P(S)$. If $i<i'$ and $A_i\cap A_{i'} \ne \emptyset$, then $M_{i'}\subseteq M_i $.
\end{lemma}
The following fact is useful in the proof of Lemma~\ref{lem:intersect-A-contain-M}.
\begin{fact}[{\cite[{Lemma 6}]{GKP19}}]\label{lem:period-of-substring}
	Let $S$ be a periodic string. If $T$ is a substring of $S$ of length at least $2\per(S)$, then $\per(S) = \per(T)$.
\end{fact}

\begin{proof}[Proof of Lemma~\ref{lem:intersect-A-contain-M}]
	Let $j\in A_i\cap A_{i'}$.
	By definition,
	$[j\dd j+2\ell)\subseteq (R_i\setminus M_i) \cap (R_{i'}\sm M_{i'})$.
	Thus, $\mu_i^*[1+j-i\dd 2\ell +j-i] =S^*[j\dd j+2\ell-1]=\mu_{i'}^*[1+j-i'\dd 2\ell +j-i']$.
	Since $\rho_i=\per(\mu_i^*)\le \ell$ and $\rho_{i'}=\per(\mu_{i'}^*)\le \ell$, by Lemma~\ref{lem:period-of-substring}, we have $\rho_i=\per(\mu_i^*)=\per(\mu_i^*[1+j-i\dd 2\ell +j-i])=\per(\mu_{i'}^*[1+j-i'\dd 2\ell +j-i'])=\rho_{i'}$.
	Therefore, $\mu_{i'}^*[1\dd \tau_{i'}]= \mu_i^*[i'-i+1\dd i'-i+\tau_{i'}]$ (since the two fragments are extensions of the same periodic string with the same period).
	Hence, for any $\tau\le\tau_{i'}$, we have
	$\HAM(S^*[i'\dd i'+\tau-1],\mu_{i'}^*[1\dd \tau])=\HAM(S^*[i'\dd i'+\tau-1],\allowbreak \mu_{i}^*[i'-i+1\dd i'-i+\tau])$.

	Since $\min (A_i\cap A_{i'})\ge i'$ and $A_i\subseteq R_i$, we have that $\tau_i> i'-i$.
	Therefore, for $\tau=i'-i$, we have $i'-i\ge \tfrac n{\gamma k} \HAM(S^*[i\dd i+i'-i-1],\mu_i^*[1\dd i'-i])= \tfrac n{\gamma k} \HAM(S^*[i\dd i'-1],\mu_i^*[1\dd i'-i])$.
	
	Thus, for any $\tau<i'-i+\tau_{i'}$, we have
	\begin{align*}
	\tfrac n{\gamma k}&\HAM\left(S^*[i\dd i+\tau-1],\mu_i^*[1\dd \tau]\right)\\&=
	\tfrac n{\gamma k}\HAM\left(S^*[i\dd i'-1],\mu_i^*[1\dd i'-i]\right)+ \tfrac{n}{\gamma k}\HAM\left(S^*[i'\dd i+\tau-1],\mu_i^*[i'-i+1\dd \tau]\right)\\&\le
	i'-i +\tfrac n{\gamma k}\HAM\left(S^*[i'\dd i'-(i'-i)+\tau-1],\mu_{i'}^*[1\dd \tau-(i'-i)]\right)\\ &\le
	i'-i +\tau-(i'-i)=\tau.
	\end{align*}
	Consequently, $\tau_i\ge i'-i+\tau_{i'}$, which means that $R_{i'}\subseteq R_i$.
	For a proof that $M_{i'}\sub M_i$, let us choose $j'\in M_{i'}$. By definition, $S[j']\ne\mu_{i'}^*[j'-i'+1]=\mu_i^*[j'-i'+1+(i'-i)]=\mu_i^*[j'-i+1]$.
	Hence, $j'\in M_i$.
\end{proof}

Lemma~\ref{lem:intersect-A-contain-M} implies that for any two indices $i<i'$, if $A_i\cap A_{i'}\ne\emptyset$, then $M_{i'}\subseteq M_i$, and thus it is enough to consider only the index $i$ when defining $\fp(S)$.
Therefore, we define $\Gamma=\{i'\in\P(S)\mid \forall i<i': A_i\cap A_{i'}=\emptyset\}$.
Notice that, among $i\in\Gamma$, all the sets $A_i$ are disjoint.
Moreover, since for any $i\in\P(S)$ we have $A_i\subseteq R_i\subseteq [1\dd 3n]$ by Lemma~\ref{lem:tau-le-2n}, we have $\sum_{i\in\Gamma}|A_i|=\left|\bigcup_{i\in\Gamma}A_i\right|\le|[1\dd 3n]|= 3n$.

For every $i\in \P(S)$, we have $|A_i| \ge |R_i|-2\ell |M_i|= |R_i|-\frac{2n}{3\gamma k}|M_i|$.
Furthermore, $|R_i|-1\ge \frac { n}{\gamma k}(|M_i|-1)$ by definition of $\tau_i = |R_i|$.
Thus, $|A_i| >\tfrac{n}{\gamma k}|M_i| - \tfrac{n}{\gamma k}-\tfrac{2n}{3\gamma k}|M_i|
= \tfrac{n}{3\gamma k}(|M_i|-3)$.
Due to $[i\dd i+\ell)\sub A_i$, we have $|A_i| \ge \ell =  \frac{n}{3\gamma k}$, and therefore
$|M_i|< \tfrac{3\gamma k}{n}|A_i|+3 \le\tfrac{3\gamma k}{n}|A_i|+ \tfrac{9\gamma k}{n}|A_i|=\tfrac{12\gamma k}{n}|A_i|$.
Hence,
$|\fp(S)| \le \Big|\bigcup_{i\in \P(S)} M_i\Big| = \Big|\bigcup_{i\in \Gamma} M_i\Big|
\le \sum_{i\in \Gamma} |M_i| \le \sum_{i\in\Gamma}\tfrac{12\gamma k}{n}|A_i|
= \tfrac{12\gamma k}{n}\sum_{i\in \Gamma}|A_i|\le 36\gamma k$.

\subparagraph*{Property~\ref{prop:rotation} of Theorem~\ref{thm:points-set}.}
The following lemma states that $\fp$ satisfies Property~\ref{prop:rotation}.

\begin{lemma}\label{lem:p-rotation}
	For every $S\in\Sigma^n$, we have $\fp(\shft(S))=\rot(\fp(S))$.
\end{lemma}
\begin{proof}
	Let $j\in\fp(\shft(S))$. There exists $i\in\P(\shft(S))$ such that $j\in M_{\shft(S),i}\cmod n$. Let $j'\in M_{\shft(S),i}$ such that $j=j'\cmod n$.
	We distinguish between two cases:
	if $i\in[1\dd n-1]$, then, since $i\in\P(\shft(S))$, we have $i+1\in\P(S)$ and $\tau_{S,i+1}=\tau_{\shft(S),i}$. Therefore, $j'+1\in M_{S,i+1}$ and $(j'+1)\cmod n\in \fp(S)$. Thus, $j=(j'+1-1)\cmod n\in\rot(\fp(S))$.
	If $i=n$, then it must be that $1\in \P(S)$ and $\tau_{S,1}=\tau_{\shft(S),n}$.
	Therefore, $j'-n+1\in M_{S,1}$ and $(j'-n+1)\cmod n\in\fp(S)$.  Thus, $j=(j'-n+1-1)\cmod n\in\rot(\fp(S))$.
	The converse inclusion holds symmetrically.
\end{proof}

\subparagraph*{Property~\ref{prop:large-intersection} of Theoren~\ref{thm:points-set}.}
We first give a lower bound on $|\fp(S)|$ in terms of $|\P(S)|$.

\begin{lemma}\label{lem:fp-large}
For every string $S\in\Sigma^n\setminus\mathcal H_{n,k}$, we have $|\fp(S)|\ge\tfrac {\gamma k}{3n}|\P(S)|$.
\end{lemma}
\begin{proof}
	First, we shall construct a set $\Delta\subseteq\P(S)$ such that $\sum_{i\in\Delta}|R_i|\ge |\P(S)|$ and, for any two distinct indices $i,i'\in\Delta$, we have $R_i\cap R_{i'}=\emptyset$.
	We build $\Delta$ iteratively.
	We start with $\Delta=\emptyset$ and, as long as $\P(S)\not \sub \bigcup_{i\in\Delta}R_i$, we add $\min\left(\P(S)\setminus \bigcup_{i\in\Delta}R_i\right)$ to $\Delta$.
	Let $i<i'$ be two indices in $\Delta$. When $i'$ was added to $\Delta$, we already had $i\in\Delta$. Thus, $R_i$ ends to the left of $i'$, which is the starting point of $R_{i'}$. Hence, $R_i\cap R_{i'}=\emptyset$.
	The algorithm terminates when $\P(S)\subseteq \bigcup_{i\in\Delta}R_i$, so $|\P(S)|\le |\bigcup_{i\in\Delta}R_i|=\sum_{i\in\Delta}|R_i|$.

	For any $i\in\P(S)$, we have $|R_i|=\tau_{i}<\tfrac { n}{\gamma k}\HAM\left(S^*[i\dd i+\tau_{i}-1],\mu_i^*[1\dd \tau_{i}]\right)$, i.e., $|R_i|<\tfrac{n}{\gamma k}|M_i|$.
	Since $M_i\subseteq R_i$ for every $i$, the sets $M_i$ for $i\in \Delta$ are disjoint.
Consequently, $|\bigcup_{i\in\Delta} M_i|= \sum_{i\in\Delta}|M_i|>\tfrac {\gamma k}{n}\sum_{i\in\Delta}|R_i|\ge\tfrac {\gamma k}n|\P(S)|$.

By Lemma~\ref{lem:tau-le-2n}, for any $i\in\P(S)$, we have $\tau_i\le 2n$. Therefore, $\bigcup_{i\in\Delta} M_i\subseteq [1\dd 3n]$ and each position in $j\in \bigcup_{i\in\Delta} (M_i\cmod n)$ may be introduced by at most $3$ positions $j,j+n,j+2n\in \bigcup_{i\in\Delta} M_i$.
Thus, $|\fp(S)|=\left|\bigcup_{i\in\Delta} (M_i\cmod n)\right|\ge \tfrac 13\left|\bigcup_{i\in\Delta} M_i\right|\ge\tfrac {\gamma k}{3n}|\P(S)|$.
\end{proof}

Using Lemma~\ref{lem:fp-large}, we prove the third property of Theorem~\ref{thm:points-set}, assuming $|\P(S_1)|\ge \frac 12n$.

\begin{lemma}\label{lem:fp-large-intersection}
	Suppose that $S_1,S_2\in \Sigma^n\sm \H_{n,k}$ satisfy $\HAM(S_1,S_2)\le k$.
	If  $|\P(S_1)|\ge \frac 12n$, then $|\fp(S_1)\cap \fp(S_2)|\ge k$.
\end{lemma}

\begin{proof}
Let $S'$ be a string of length $n$, where, for any $i$ with $S_1[i]=S_2[i]$, we have $S'[i]=S_1[i]$ and, for any other $i$ (i.e., for $i\in \MP(S_1,S_2)$), we have $S'[i]=\$_i$, where $\$_i\notin \Sigma$ differs from any other character $\$_{i'}$ for $i'\ne i$.

\begin{claim}\label{clm:fpsprime-almost-subset}
$\fp(S')\subseteq \left(\fp(S_1)\cap \fp(S_2)\right)\cup \MP(S_1,S_2)$.
\end{claim}
\begin{proof}
	Let $j\in\fp(S')$.
	If $j\in \MP(S_1,S_2)$, the claim follows; thus, assume $j\notin \MP(S_1,S_2)$.
	By the definition of $\fp(S')$, there is an index $i\in\P(S')$ such that $j\in M_{S',i}\cmod n$; let $j'\in M_{S',i}$ be an integer such that $j=j'\cmod n$.
	Notice that $\mu_{S_1,i}=\mu_{S',i}$ since if $\mu_{S',i}$ contains some $\$_k$ character, then $i$ cannot be cubic and so $i\notin\P(S')$.
	Therefore, $\mu_{S_1,i}=\mu_{S',i}$, and let $\mu_i=\mu_{S_1,i}$.
	For any integer $\tau$, we have $\HAM(S_1^*[i\dd i+\tau-1],\mu_i^*[1\dd \tau])\le \HAM((S')^*[i\dd i+\tau-1],\mu_i^*[1\dd \tau])$ because the new $\$_k$ characters in $S'$ just form new mismatches.
	In particular, for $\tau_{S_1,i}$ we have $\frac { n}{\gamma k}\HAM((S')^*[i\dd i+\tau_{S_1,i}-1],\mu_i^*[1\dd \tau_{S_1,i}])\ge \frac { n}{\gamma k}\HAM(S_1^*[i\dd i+\tau_{S_1,i}-1],\mu_i^*[1\dd \tau_{S_1,i}])>\tau_{S_1,i}$.
	Hence, $\tau_{S',i}\le\tau_{S_1,i}$ and $R_{S',i}\subseteq R_{S_1,i}$.
	Since $j'\in M_{S',i}$ and $j\notin \MP(S_1,S_2)$, it must be that $j'\in M_{S_1,i}$.
	Similarly, $j'\in M_{S_2,i}$. Thus, $j=j'\cmod n\in \left(\fp(S_1)\cap\fp(S_2)\right)\cup \MP(S_1,S_2)$.
\end{proof}

\begin{claim}\label{clm:psprime-large}
	$|\P(S')|\ge \frac{\gamma-2}{2\gamma}n$.
\end{claim}
\begin{proof}
	Recall that $|\P(S_1)|\ge\frac 12n$.
	If $\mu_{S_1,i}=\mu_{S',i}$ and $i\in\P(S_1)$, then $i\in\P(S')$.
	The only indices $i\in \P(S_1)\cap \NP(S')$ are indices such that $\mu_{S_1,i}\ne\mu_{S',i}$, which means that $\MP(S_1,S_2)\cap([i\dd i+3\ell-1]\cmod n)\ne\emptyset$.
	Hence, each $m\in \MP(S_1,S_2)$ will remove at most $3\ell$ positions from $\P(S_1)$.
	Thus, $|\P(S')|\ge\frac 12n-|\MP(S_1,S_2)|3\ell\ge \frac 12n-k\frac n{\gamma k}= \frac{\gamma-2}{2\gamma}n$.
\end{proof}

Due to Claim~\ref{clm:psprime-large}, we have $|\P(S')|\ge  \frac{\gamma-2}{2\gamma}n$, and therefore $|\fp(S')|>\frac{ \gamma k}{3n} \frac{\gamma-2}{2\gamma}n=\frac {(\gamma-2)k}{6}$ by Lemma~\ref{lem:fp-large}.
Due to Claim~\ref{clm:fpsprime-almost-subset}, $\fp(S')\subseteq \left(\fp(S_1)\cap \fp(S_2)\right)\cup \MP(S_1,S_2)$, and therefore $|\fp(S')|\le |\left(\fp(S_1)\cap \fp(S_2)\right)\cup \MP(S_1,S_2)|\le |\fp(S_1)\cap \fp(S_2)|+|\MP(S_1,S_2)|\le |\fp(S_1)\cap \fp(S_2)|+k$.
Consequently, since $\gamma \ge 14$, we have $|\fp(S_1)\cap \fp(S_2)|\ge \frac{\gamma-8}{6}k  \ge \frac{14-8}6k=k$.
\end{proof}

\section{Sketches for Pseudo-periodic Strings}\label{sec:ap-sketches}
Let $\H'_{n,k}\sub \Sigma^n$ be the family of $(3\gamma k,(\gamma+1) k)$-pseudo-periodic strings in $\Sigma^n$.
In this section, we develop circular sketches for $\H'_{n,k}$. 
We start with a few properties of pseudo-periodic strings.
Recall that a string $S\in \Sigma^n$ is called \emph{$(\alpha, \beta)$-pseudo-periodic} if it has an $(\alpha,\beta)$-base $S'\in \Sigma^n$ with $\prr(S')\le \frac{n}{\alpha}$ and $\HAM(S,S')\le \beta$.
If $\floor{\alpha}>2\beta$, then the  $(\alpha,\beta)$-base is unique.
\begin{lemma}\label{lem:highly-periodic-same-root}
	If $S\in\Sigma^n$ is an $(\alpha,\beta)$-pseudo-periodic string for some parameters $\floor{\alpha}>2\beta$,
	then $S'$ has a unique $(\alpha,\beta)$-base.
\end{lemma}
\begin{proof}
	Suppose that $S$ has two bases $S',S''$.
	Alzamel et al.~\cite{ACIKKRRW18} show that if $|X|=|Y|\ge \per(X)+\per(Y)$ and $X\ne Y$, then $\HAM(X,Y)\ge \big\lfloor{\frac{2n}{\per(X)+\per(Y)}}\big\rfloor$.
	Setting $X=S'$ and $Y=S''$, we get a contradiction: $\HAM(S',S'')\ge \big\lfloor{\frac{2n}{\per(S')+\per(S'')}}\big\rfloor\ge  \big\lfloor{\frac{2n}{\prr(S')+\prr(S'')}} \big\rfloor\ge \big\lfloor{\frac{2n}{n/\alpha+n/\alpha}\big\rfloor} =\floor{\alpha}> 2\beta \ge   \HAM(S,S')+\HAM(S,S'') \ge \HAM(S',S'')$.
\end{proof}

Moreover, the triangle inequality immediately yields the following observation.

\begin{observation}\label{obs:ap}
	Let $S\in\Sigma^n$ be an $(\alpha,\beta)$-pseudo-periodic string and let $T\in\Sigma^n$ be a string such that $\HAM(S,T)\le k$. Then, $T$ is $(\alpha,\beta+k)$-pseudo-periodic, and every $(\alpha,\beta)$-base of $S$ is an $(\alpha,\beta+k)$-base of $T$.
\end{observation}

Combining Lemma~\ref{lem:highly-periodic-same-root} with Observation~\ref{obs:rootrot} and Observation~\ref{obs:ap}, we obtain the following corollary.
\begin{corollary}\label{cor:sameroot}
	Let $S_1,S_2\in\H'_{n,k}$ with $(3\gamma k,(\gamma+1) k)$-bases $S'_1$ and $S'_2$, respectively.
	If, for some $m\in \mathbb Z$, we have $\HAM(S_1,\shft^m(S_2))\le k$, then $S'_1 = \shft^m(S'_2)$.
\end{corollary}
\begin{proof}
	By Observation~\ref{obs:ap}, $S'_1$ is a $(3\gamma k,(\gamma+2) k)$-base of  $\shft^m(S_2)$.
	Moreover, by Observation~\ref{obs:rootrot}, $\shft^m(S'_2)$ is a $(3\gamma k,(\gamma+1) k)$-base of $\shft^m(S_2)$, and thus also a $(3\gamma k,(\gamma+2) k)$-base of  $\shft^m(S_2)$.
	Since $\floor{3\gamma k} > 2(\gamma+2)k$ due to $\gamma \ge 5$,
	Lemma~\ref{lem:highly-periodic-same-root} implies that $S'_1 =\shft^m(S'_2)$.
\end{proof}

\subsection{A $0$-mismatch Circular Sketch}
Both the exact and the $(1\pm\eps)$-approximation sketches of strings in $\H'_{n,k}$ rely on $0$-mismatch circular sketches, which we implement using Karp--Rabin fingerprints.

\begin{fact}[Karp--Rabin fingerprints~\cite{KR87}]\label{fct:kr-fingerprint}
	For every positive integer $n$, there exists a randomized function $\Phi : \Sigma^{n} \to \{0,1\}^{\Oh(\log n)}$
	such that, for every $S_1,S_2\in \Sigma^{n}$, the following holds with high probability:
	if $S_1 \ne S_2$, then $\Phi(S_1)\ne \Phi(S_2)$.
\end{fact}
\begin{proof}
	The function $\Phi$ is based on a fixed prime number $p \ge \max(\sigma,n^2)$
	and a uniformly random $x\in [0\dd p-1]$.
	The function $\Phi$ maps a string $S$ to $(\sum_{i=1}^{|S|} x^{i-1}\cdot S[i]) \bmod p$.
	This way, for every two strings $S_1\ne S_2$ in $\Sigma^n$,
	we have $\Pr[\Phi(S_1)= \Phi(S_2)] \le \frac{n}{p} \le \frac{n}{n^2}=n^{-1}$.
\end{proof}

\begin{lemma}\label{lem:0}
	There exists a $0$-ECS sketch $(\sk_0,\dec_0)$ for $\Sigma^n$ of size $\Oh(\log n)$ bits
	with constant decoding time.
\end{lemma}
\begin{proof}
	The construction relies on a Karp--Rabin fingerprint function $\Phi$.
	The sketch $\sk_0(S)$ for a string $S\in \Sigma^n$ is defined based on the minimum cyclic rotation of $S$,
	denoted $\minrot(S)$, and consists of the following components:
	\begin{itemize}
		\item the fingerprint $\Phi(\minrot(S))$ of the minimum cyclic rotation of $S$,
		\item the length $\prr(S)$ of the primitive root of $S$,
		\item the smallest integer $r\ge 0$ such that $S=\shft^r(\minrot(S))$.
	\end{itemize}
	The decoding function $\dec_0$ is given two sketches $\sk_0(S_1)=(\Phi(\minrot(S_1)),\prr(S_1),r_1)$, $\sk_0(S_2)=(\Phi(\minrot(S_2)),\prr(S_2),r_2)$, and a shift $m$. If $\Phi(\minrot(S_1))\ne \Phi(\minrot(S_2))$, then $S_1 \ne \shft^m(S_2)$, and thus the function returns $\infty$.
	Otherwise, $\minrot(S_1)=\minrot(S_2)$ with high probability, and the implementation proceeds
	assuming that $\minrot(S_1)=T=\minrot(S_2)$ for a string $T\in \Sigma^n$.
	In particular, this implies $\prr(S_1)=\prr(T)=\prr(S_2)$.
	Finally, since $S_1 = \shft^{r_1}(T)$ equals $\shft^m(S_2)=\shft^{m+r_2}(T)$ if and only if $\prr(T) \mid (m + r_2-r_1)$, the function returns $0$ or $\infty$ depending on whether $\prr(S_1) \mid (m + r_2-r_1)$ or not.
\end{proof}

\subsection{A $k$-ECS Sketch}

\begin{construction}\label{def:circksk-ap}
	The encoding function $\circ_{k} : \H'_{n,k} \to \{0,1\}^*$ is defined as follows:
	\begin{enumerate}
		\item Let $\sk_0$ be the $0$-mismatch sketch of Lemma~\ref{lem:0}.
		\item For $S\in \H'_{n,k}$, the encoding $\circ_k(S)$ stores the sketch $\sk_0(S')$ of the $(3\gamma k, (\gamma+1)k)$-base $S'$ of $S$ and the mismatch information $\MI(S, S')$.
	\end{enumerate}
\end{construction}

\begin{proposition}\label{prp:exact-ap-sketch}
	There exists a decoding function which, together with the encoding $\circ_k$ of Definition~\ref{def:circksk-ap},
	forms a $k$-ECS sketch of $\H'_{n,k}$.
	The size of the sketch is $\Ohtilde(k)$, and the decoding time is $\Ohtilde(k)$ with high probability.
\end{proposition}
\begin{proof}
	The decoding function is given two sketches $\circ_k(S_1)=(\sk_0(S'_1),\MI(S_1, S'_1))$, $\circ_k(S_2)=(\sk_0(S'_2),\MI(S_2, S'_2))$, and a shift $m$.
	By Corollary~\ref{cor:sameroot}, if $\HAM(S_1,\shft^m(S_2))\le k$,
	then $S'_1 = \shft^m(S'_2)$, and this condition is checked by applying $\dec_0(\sk_0(S'_1),\sk_0(S'_2),m)$. If the call returns a non-zero result, then $\infty$ is returned.
	Otherwise,  $S'_1 = \shft^m(S'_2)$ holds with high probability.
	The analysis below is conditioned on this event.

	First, $\MI(\shft^m(S_2),\shft^m(S'_2))$ is retrieved from $\MI(S_2,S'_2)$
	by shifting all the the mismatches.
	Next, the decoding function retrieves $\MI(S_1,\shft^m(S_2))$ from $\MI(S_1, S'_1)$
	and  $\MI(\shft^m(S_2),\shft^m(S'_2))$ (using Fact~\ref{fct:retrieve} and assuming that $S'_1 = \shft^m(S'_2)$) and returns $\HAM(S_1,\shft^m(S_2))=|\MI(S_1,\shft^m(S_2))|$.
\end{proof}

\subsection{An $(\eps,k)$-ACS Sketch}

For the pseudo-periodic $(\eps,k)$-ACS sketches, we relax the problem statement;
we overcome this relaxation in Section~\ref{sec:summary}.
In the \emph{relaxed} $(\eps,k)$-ACS sketch, the distances smaller than $\frac{k}{2}$ do not need to be approximated. More precisely, we require the following:
\begin{itemize}
	\item if $\HAM(S_1,\shft^m(S_2)) < \frac12k$, then $\dec(\sk(S_1),\sk(S_2),m) < \frac{1+\eps}2k$,
	\item if $\frac12k \le \HAM(S_1,\shft^m(S_2)) \le k$, then $\dec(\sk(S_1),\sk(S_2),m)\in (1\pm\eps)\HAM(S_1,\shft^m(S_2))$,
	\item otherwise, $\dec(\sk(S_1),\sk(S_2),m)>(1-\eps)k$.
\end{itemize}

\begin{construction}\label{def:circepssk-ap}
	The encoding function $\circ_{\eps,k} : \H'_{n,k} \to \{0,1\}^*$ is defined as follows:
	\begin{enumerate}
		\item Let $\sk_0$ be the $0$-mismatch sketch of Lemma~\ref{lem:0}.
		\item Let $A,B\sub [n]$ be two subsets with elements sampled independently with rate $p:=\sqrt{\frac{\log n}{\eps^2 k}}$.
		\item For $S\in \H'_{n,k}$, the encoding $\circ_{\eps,k}(S)$ stores the sketch $\sk_0(S')$ of the $(3\gamma k, (\gamma+1)k)$-base $S'$ of $S$ and the mismatch information $\MI_{A\cup B}(S, S')$.
	\end{enumerate}
\end{construction}

\begin{proposition}\label{prp:approx-ap-sketch}
	There exists a decoding function which, together with the encoding $\circ_{\eps,k}$ of Definition~\ref{def:circepssk-ap},
	forms a relaxed $(\eps,k)$-ACS sketch of $\H'_{n,k}$.
	The size of the sketch is $\Ohtilde(\eps^{-1}\sqrt{k})$, and the decoding time is $\Ohtilde(\eps^{-1}\sqrt{k})$  with high probability.
\end{proposition}
\begin{proof}
	The decoding function is given two sketches $\circ_{\eps,k}(S_1)=(\sk_0(S'_1),\MI_{A\cup B}(S_1, S'_1))$ and $\circ_{\eps,k}(S_2)=(\sk_0(S'_2),\MI_{A\cup B}(S_2, S'_2))$, and a shift $m$.
	According to Corollary~\ref{cor:sameroot}, if $\HAM(S_1,\shft^m(S_2))\le k$,
	then $S'_1 = \shft^m(S'_2)$, and this condition is checked by applying $\dec_0(\sk_0(S'_1),\sk_0(S'_2),m)$. If the call returns a non-zero result, then $\infty$ is returned.
	Otherwise,  $S'_1 = \shft^m(S'_2)$ holds with high probability.
	The analysis below is conditioned on this event.

	First, $\MI_{A\cap \rot^m(B)}(\shft^m(S_2),\shft^m(S'_2))$ is retrieved by filtering and shifting $\MI_{A\cup B}(S_2,S'_2)$.
	Secondly, $\MI_{A \cap \rot^m(B)}(S_1,S'_1)$ is retrieved by filtering  $\MI_{A\cup B}(S_1,S'_1)$.
	Then, the algorithm retrieves $\MI_{A\cap \rot^m(B)}(S_1,\shft^m(S_2))$ combining $\MI_{A\cap \rot^m(B)}(S_1,S'_1)$ and $\MI_{A\cap \rot^m(B)}(\shft^m(S_2),\allowbreak\shft^m(S'_2))$ (using Fact~\ref{fct:retrieve} and assuming that $S'_1 = \shft^m(S'_2)$).
	Since $A\cap \rot^m(B)$ is a random subset of $[n]$ with elements sampled independently with rate $\frac{\log n}{\eps^2 k}$, the quantity $\frac{\eps^2 k}{\log n}\HAM_{A\cap \rot^m(B)}(S_1,\shft^m(S_2))$
	 is a $(1\pm\eps)$-approximation of $\HAM(S_1,S_2)$	with high probability provided that $\HAM(S_1,S_2)=\Omega(k)$; see Lemma~\ref{lem:chernoff}.
\end{proof}

\section{Proofs of Main Theorems}\label{sec:summary}
In this section, we complete our construction of circular $k$-mismatch sketches for $\Sigma^n$.

\thmexact*
\begin{proof}
Our construction combines the $k$-ECS sketches of Theorem~\ref{thm:exact-nap-sketch} and Property~\ref{prp:exact-ap-sketch}.
For each string $S\in \Sigma$, if $S\in\H'_{n,k}$, then the sketch contains the sketch of Property~\ref{prp:exact-ap-sketch}, and if $S\in \Sigma^n\setminus\H_{n,k}$, then the sketch contains the sketch of Theorem~\ref{thm:exact-nap-sketch}. Notice that the sketch contains both components if  $S\in\H'_{n,k}\sm\H_{n,k}$.

For two strings $S_1,S_2\in\Sigma^n$, given the sketches of $S_1$  and $S_2$, the decoder works as follows.
If the two sketches contain compatible components (of Theorem~\ref{thm:exact-nap-sketch} or of Property~\ref{prp:exact-ap-sketch}), then the decoder uses the decoder corresponding to these components.
Otherwise, without loss of generality, it must be that $S_1\in\H_{n,k}$ and $S_2\notin \H'_{n,k}$. Thus, by Observartion~\ref{obs:ap}, $\HAM(S_1,S_2)>k$, and therefore the decoder outputs $\infty$.
The decoding time is $\Ohtilde(k)$.
\end{proof}

Similarly, combining the results of Section~\ref{sec:nap} and Section~\ref{sec:ap-sketches} gives $(1+\eps)$-approximate sketches.
The proof of the following result mimics the proof of Theorem~\ref{thm:exact-circular-sketch}.

\begin{proposition}\label{lem:approx-circular-sketch-relaxed}
	There exists a relaxed $(\eps,k)$-ACS sketch for $\Sigma^n$ of size $\Ohtilde(\eps^{-2}\sqrt k)$. Its decoding time is $\Ohtilde(\eps^{-1}\sqrt k+\eps^{-2})$  with high probability.
\end{proposition}

\begin{proof}
	Our construction combines the $(\eps,k)$-ACS sketch of Theorem~\ref{thm:approx-nap-circ-sketch} and the relaxed $(\eps,k)$-ACS sketch of Property~\ref{prp:approx-ap-sketch}.
	For each strings $S$, if $S\in\H'_{n,k}$, then the sketch contains the sketch of $S$ by Property~\ref{prp:approx-ap-sketch} and, if $S\in \Sigma^n\setminus\H_{n,k}$, then the sketch contains the sketch of $S$ by Theorem~\ref{thm:approx-nap-circ-sketch}.
	Notice that, for  $S\in\H'_{n,k}\sm\H_{n,k}$  the sketch contains both components.
	
	For any two strings $S_1,S_2\in\Sigma^n$, given the sketches of $S_1$  and $S_2$, the decoder works as follows.
	If the two sketches contains compatible components (of Theorem~\ref{thm:approx-nap-circ-sketch} or of Property~\ref{prp:approx-ap-sketch}), then the decoder uses the decoder corresponding to these components.
	Otherwise, without loss of generality, it must be that $S_1\in\H_{n,k}$ and $S_2\notin \H'_{n,k}$. Thus, by Observation~\ref{obs:ap}, $\HAM(S_1,S_2)>k$, and therefore the decoder outputs $\infty$.
\end{proof}

A simple alternative approach yields smaller sketches when $k$ is large compared to $n$.
\begin{construction}\label{def:alternative}
	The encoding function $\circ_{\eps,k} : \Sigma^n \to \{0,1\}^*$ is defined as follows:
	\begin{enumerate}
		\item Let $A,B\sub [n]$ be two subsets with elements sampled independently with rate $p:=\sqrt{\frac{\log n}{\eps^2 k}}$.
		\item For $S\in \Sigma^n$, the encoding $\circ_{\eps,k}(S)$ consists of pairs $(i,S[i])$ for $i\in A\cup B$.
	\end{enumerate}
\end{construction}
\begin{proposition}\label{prp:approx-circular-sketch-alternative}
	There exists a decoding function which, together with the encoding $\circ_{\eps,k}$ of Definition~\ref{def:alternative},
	forms a relaxed $(\eps,k)$-ACS sketch of $\Sigma^n$.
	The size of the sketch is $\Ohtilde(\frac{n}{\eps \sqrt{k}})$, and the decoding time is $\Ohtilde(\frac{n}{\eps \sqrt{k}})$  with high probability.
\end{proposition}
\begin{proof}
	The decoder, given the sketches of $S_1,S_2\in \Sigma^n$ and a shift $m$,
	uses Lemma~\ref{lem:chernoff} to estimate $\HAM(S_1,\shft^m(S_2))$ based on $\HAM_{A \cap \rot^m(B)}(S_1,\shft^m(S_2))$. For each $i\in A \cap \rot^m(B)$, the decoder retrieves $S_1[i]$ from the sketch of $S_1$
	and $\shft^m(S_2)[i]=S_2[(i-m)\cmod n]$ from the sketch of $S_2$.
	Since $A\cap \rot^m(B)$ is a random subset of $[n]$ with elements sampled independently with rate $\frac{\log n}{\eps^2 k}$, the quantity $\frac{\eps^2 k}{\log n}\HAM_{A\cap \rot^m(B)}(S_1,\shft^m(S_2))$
	 is a $(1\pm\eps)$-approximation of $\HAM(S_1,\shft^m(S_2))$	with high probability provided that $\HAM(S_1,\shft^m(S_2))=\Omega(k)$; see Lemma~\ref{lem:chernoff}.
\end{proof}

\thmappx*

\begin{proof}
An $(\eps,k)$-ACS  sketch is obtained by combining $\Oh(\log k)$ relaxed $(\eps,k')$-ACS sketches, where $k'$ ranges over powers of two between $1$ and $2k$.  Depending on whether $k'\le \eps n$ or not, Lemma~\ref{lem:approx-circular-sketch-relaxed} or Property~\ref{prp:approx-circular-sketch-alternative} is used to implement $k'$-mismatch sketches. 
\end{proof}

\begin{remark}
Applying Property~\ref{prp:approx-circular-sketch-alternative} instead of Lemma~\ref{lem:approx-circular-sketch-relaxed} improves the sketch size (for $k\ge \eps n$) but degrades the decoding time.
We get two alternatives: $\Ohtilde(\eps^{-2}\sqrt{k})$-size sketches with decoding time $\Ohtilde(\eps^{-1}\sqrt{k}+\eps^{-2})$, and $\Ohtilde(\eps^{-1.5}\sqrt{n})$-size sketches with decoding time $\Ohtilde(\eps^{-1.5}\sqrt{n})$.
\end{remark}

\section{Efficient Shift Distance Decoders}\label{sec:shift-decoding}
In this section, we develop exact and approximate $k$-mismatch shift distance sketches
with efficient decoding procedures.
These sketches use the same encoding functions as the corresponding $k$-mismatch 
circular sketches, so we only need to develop the decoding procedures.

Our decoding procedures for shift distance heavily rely on their counterparts for decoding
the Hamming distance between $S_1$ and a fixed rotation of $S_2$. Hence, each of the following 
four propositions refers to its counterpart in Section~\ref{sec:nap} or Section~\ref{sec:ap-sketches}.

\subsection{Shift Distance Sketches for Non-Pseudo-Periodic Strings}
\begin{proposition}[see Theorem~\ref{thm:approx-nap-circ-sketch}]\label{thm:approx-nap-shift-sketch}
	There exists a decoding function which, together with the encoding $\circ_{\varepsilon,k}$  of Definition~\ref{def:sketchesapproxcirc}, forms an $(\eps,k)$-ASDS sketch of $\Sigma^n \sm \H_{n,k}$. The decoding algorithm costs $\Ohtilde(\eps^{-2}k)$ time with high probability.
\end{proposition}
\begin{proof}
	Our decoding procedure iterates over $i \in f(S_1)\cap A$ and $i'\in f(S_2)\cap B$.
	For each such pair $(i,i')$, the procedure retrieves the sketches $\sk_{\varepsilon}(\shft^{i}(S_1))$ and $\sk_{\varepsilon}(\shft^{i'}(S_2))$ and recovers a $(1+\eps)$-approximation of $\HAM(\shft^i(S_1),\shft^{i'}(S_2))$. The algorithm returns the smallest among the values obtained across
	all the iterations.
	
	Since $\HAM(\shft^i(S_1),\shft^{i'}(S_2))\ge \sh(S_1,S_2)$, the returned value is at least $(1-\eps)\sh(S_1,S_2)$ with high probability (unless the sketches $\sk_{\varepsilon}$ fail).
	Moreover, if $\sh(S_1,S_2) \le k$ with $\HAM(S_1,\shft^m(S_2))=\sh(S_1,S_2)$ for some integer $m$,
	then, as in the proof of Theorem~\ref{thm:approx-nap-circ-sketch}, with high probability, there is a pair of indices
	$i\in f(S_1)\cap A$ and $i'\in f(S_2)\cap B$ with $i' = (i+m)\cmod n$.
	Hence, the returned value is at most $(1+\eps)\HAM(S_1,\shft^m(S_2))=(1+\eps)\sh(S_1,S_2)$
	with high probability.
\end{proof}

\begin{proposition}[see Theorem~\ref{thm:exact-nap-sketch}]\label{thm:exact-nap-shift-sketch}
	There exists a decoding function which, together with the encoding $\circ_{k}$  of Definition~\ref{def:sketchesexactcirc}, forms a $k$-ESDS sketch of $\Sigma^n \sm \H_{n,k}$.
	The decoding algorithm costs $\Ohtilde(k^2)$ time with high probability.
\end{proposition}
\begin{proof}
	The decoding algorithm first  computes the sizes
	 $s_m:= |f(S_1)\cap \rot^m(f(S_2))|$ for all shifts $m\in [n]$.
	For this, the algorithm iterates over $i\in f(S_1)$ and $i'\in f(S_2)$
	incrementing $s_{(i'-i)\cmod n}$.
	Next, for each shift $m\in [n]$ with $c_m \ge k$, the algorithm uses the decoding function of Theorem~\ref{thm:exact-nap-sketch} to retrieve $\HAM(S_1,\shft^m(S_2))$ (or learn that $\HAM(S_1,\shft^m(S_2))>k$).
	Finally, the algorithm returns the smallest among the reported values.
	(If $c_m < k$ for each $m\in [n]$, then the algorithm returns $\infty$.)

	As for correctness, first note that $\HAM(S_1,\shft^m(S_2))\ge \sh(S_1,S_2)$ holds for each $m\in [n]$, so the returned value is at least $\min(k+1,\sh(S_1,S_2))$ with high probability (unless the decoding procedure of Theorem~\ref{thm:exact-nap-sketch} fails).
	Next, suppose that $\sh(S_1,S_2)=\HAM(S_1,\shft^m(S_2))\le k$.
	As argued in the proof of Theorem~\ref{thm:exact-nap-sketch}, $s_m=|f(S_1)\cap f(\shft^m(S_1))|\ge k$
	holds with high probability. Consequently, the decoding procedure of Theorem~\ref{thm:exact-nap-sketch} was called for $S_1$, $S_2$, and $m$, resulting in $\sh(S_1,S_2)$ with high probability.
	Hence, the returned value is at most $\sh(S_1,S_2)$ with high probability.
	
	The decoder iterates over $f(S_1)\times f(S_2)$, which is of size $\Ohtilde(k^2)$ with high probability due to Theorem~\ref{thm:points-set}.
	Hence, by the pigeonhole principle there are at most $\Ohtilde(\frac {k^2}k)=\Ohtilde(k)$ positions $m\in[n]$ such that $c_m\ge k$.
	For each such position, the decoding time of Theorem~\ref{thm:exact-nap-sketch} is $\Ohtilde(k)$. Thus, the total decoding time is $\Ohtilde(k^2)$.
\end{proof}

\subsection{Shift Distance Sketches for Pseudo-Periodic Strings}
\begin{lemma}[see Lemma~\ref{lem:0}]\label{lem:0-shift}
There exists a decoding function $\shdec_0$ which, together with the encoding $\sk_0$ of Lemma~\ref{lem:0},
forms an exact $0$-ESDS sketch with constant decoding time.
\end{lemma}
\begin{proof}
The decoding function, given the sketches  $\sk_0(S_1)=(\Phi(\minrot(S_1)),\prr(S_1),r_1)$ and $\sk_0(S_2)=(\Phi(\minrot(S_2)),\prr(S_2),r_2)$,
returns $0$ or $\infty$ based on whether $\Phi(\minrot(S_1))=\Phi(\minrot(S_2))$ or not.
\end{proof}

\begin{proposition}[see Property~\ref{prp:exact-ap-sketch}]\label{prp:exact-ap-shift-sketch}
	There exists a decoding function which, together with the encoding $\circ_{k}$  of Definition~\ref{def:circksk-ap}, forms a $k$-ESDS sketch of $\H'_{n,k}$.
	The decoding algorithm costs $\Ohtilde(k^2)$ time with high probability.
\end{proposition}
\begin{proof}
	The decoding algorithm is given the sketches $\circ_k(S_1)=(\sk_0(S'_1),\MI(S_1, S'_1))$ and $\circ_k(S_2)=(\sk_0(S'_2),\MI(S_2, S'_2))$.
	First, the algorithm applies $\shdec_0(\sk_0(S'_1),\sk_0(S'_2))$ of Lemma~\ref{lem:0-shift}.
	If this call returns a non-zero result, then $\infty$ is returned.
	Otherwise, for each $m\in [n]$, the algorithm constructs the following sets:
	\begin{align*}
		\PM_m &:= \MP(S_1,S'_1)\cap \MP(\shft^m(S_2),\shft^m(S'_2))\\
		\PM'_m &:= \PM_m \sm \MP(S_1,\shft^m(S_2))
	\end{align*}
	For this, the algorithm iterates over $(i,a,b)\in \MI(S_1,S'_1)$ and $(i',c,d)\in \MP(S_2,S'_2)$,
	adding $i$ to  $\PM_{(i'-i)\cmod n}$ and, provided that $a= c$, also to $\PM'_{(i'-i)\cmod n}$.

	For each shift $m$ with $\PM_m \ne \emptyset$,
	the algorithm uses $\dec_0(\sk_0(S'_1),\sk_0(S'_2),m)$ of Lemma~\ref{lem:0}.
	If this call returns a non-zero result, then $m$ is ignored.
	Otherwise, $\HAM(S_1,\shft^m(S_2))=\HAM(S_1,S'_1)+\HAM(S_2,S'_2)-|\PM_m|-|\PM'_m|$ is computed.
	Finally, the algorithm returns the minimum of $\HAM(S_1,S'_1)+\HAM(S_2,S'_2)$ and the smallest among the computed values $\HAM(S_1,\shft^m(S_2))$.

	\subparagraph*{Correctness.}
	By Corollary~\ref{cor:sameroot}, $\sh(S_1,S_2)\le k$ guarantees $\sh(S'_1,S'_2)=0$,
	so the algorithm correctly returns $\infty$ if $\shdec_0(\sk_0(S'_1),\sk_0(S'_2))$ yields a non-zero result.
	Moreover, $\HAM(S_1,\shft^m(S_2))\le k$ guarantees $S'_1 = \shft^m(S'_2)$,
	so the algorithm correctly ignores $m\in [n]$ if $\dec_0(\sk_0(S'_1),\sk_0(S'_2),m)$ yields a non-zero result.
	In the following, we assume $\sh(S'_1,S'_2)=0$ with $S'_1 = \shft^m(S'_2)$ for all the shifts considered.
	The latter assumption implies $\HAM(S_1,\shft^m(S_2))=\HAM(S_1,S'_1)+\HAM(S_2,S'_2)-|\PM_m|-|\PM'_m|$
	(compare the proof of Fact~\ref{fct:retrieve}).
	Moreover, $\sh(S_1,S_2)\le \HAM(S_1,S'_1)+\HAM(S_2,S'_2)$ holds by the triangle inequality,
	Hence, the returned value is at least $\sh(S_1,S_2)$ with high probability.

	On the other hand, if $\sh(S_1,S_2)=\HAM(S_1,\shft^m(S_2))\le k$ for some shift $m\in [n]$,
	then $S'_1 = \shft^m(S'_2)$ and  $\HAM(S_1,\shft^m(S_2))=\HAM(S_1,S'_1)+\HAM(S_2,S'_2)-|\PM_m|-|\PM'_m|$.
	This either yields $\sh(S_1,S_2)= \HAM(S_1,S'_1)+\HAM(S_2,S'_2)$ (in case of $\PM_m=\emptyset$, which yields $|\PM_m|=|\PM'_m|=0$)
	or that $m$ was among the shifts considered (otherwise).
	In both cases, we conclude that the returned value is at most $\sh(S_1,S_2)$ with high probability.
\end{proof}

A relaxed $(\eps,k)$-ASDS sketch is defined analogously to a relaxed $(\eps,k)$-ACS sketch:
\begin{itemize}
	\item if $\sh(S_1,S_2) < \frac12k$, then $\shdec(\sk(S_1),\sk(S_2)) < \frac{1+\eps}2k$,
	\item if $\frac12k \le \sh(S_1,S_2) \le k$, then $\shdec(\sk(S_1),\sk(S_2))\in (1\pm\eps)\sh(S_1,S_2)$,
	\item otherwise, $\shdec(\sk(S_1),\sk(S_2))>(1-\eps)k$.
\end{itemize}

\begin{proposition}[see Property~\ref{prp:approx-ap-sketch}]\label{prp:approx-ap-shift-sketch}
	There exists a decoding function which, together with the encoding $\circ_{\eps,k}$  of Definition~\ref{def:circepssk-ap}, forms a relaxed $(\eps,k)$-ASDS sketch of $\H'_{n,k}$.
	The decoding algorithm costs $\Ohtilde(\eps^{-2} k)$ time with high probability.
\end{proposition}
\begin{proof}
	The decoding algorithm is given sketches $\circ_{\eps,k}(S_1)=(\sk_0(S'_1),\MI_{A\cup B}(S_1, S'_1))$ and $\circ_{\eps,k}(S_2)=(\sk_0(S'_2),\MI_{A\cup B}(S_2, S'_2))$.
	First, the algorithm applies $\shdec_0(\sk_0(S'_1),\sk_0(S'_2))$ of Lemma~\ref{lem:0-shift}.
	If this call returns a non-zero result, then $\infty$ is returned.
	Otherwise, for each $m\in [n]$, the algorithm constructs the sets $\PM_m \cap A \cap \rot^m(B)$
	and $\PM'_m \cap A \cap \rot^m(B)$, where 	
	\begin{align*}
		\PM_m &:= \MP(S_1,S'_1)\cap \MP(\shft^m(S_2),\shft^m(S'_2))\\
		\PM'_m &:= \PM_m \sm \MP(S_1,\shft^m(S_2))
	\end{align*}
	are defined as in the proof of Property~\ref{prp:exact-ap-shift-sketch}.
	For this, the algorithm iterates over $(i,a,b)\in \MI_{A}(S_1,S'_1)$ and $(i',c,d)\in \MI_{B}(S_2,S'_2)$, adding $i$ to $\PM_{(i'-i)\cmod n} \cap A \cap \rot^{(i'-i)\cmod n}(B)$ and, provided that $a=c$, also to $\PM'_{(i'-i)\cmod n} \cap A \cap \rot^{(i'-i)\cmod n}(B)$.

	For each shift $m$ with $\PM_m \cap A \cap \rot^m(B) \ne \emptyset$,
	the algorithm uses $\dec_0(\sk_0(S'_1),\sk_0(S'_2),m)$ of Lemma~\ref{lem:0}.
	If this call returns a non-zero result, then $m$ is ignored.
	Otherwise, the algorithm computes \[d_m:=\HAM(S_1,S'_1)+\HAM(S_2,S'_2)-\tfrac{\eps^2 k}{\log n}(|\PM_m \cap A \cap \rot^m(B)|+|\PM'_m \cap A \cap \rot^m(B)|).\]
	Finally, the algorithm returns the minimum of $\HAM(S_1,S'_1)+\HAM(S_2,S'_2)$ and the smallest among the computed values $d_m$.

	\subparagraph*{Correctness.}
	By Corollary~\ref{cor:sameroot}, $\sh(S_1,S_2)\le k$ guarantees $\sh(S'_1,S'_2)=0$,
	so the algorithm correctly returns $\infty$ if $\shdec_0(\sk_0(S'_1),\sk_0(S'_2))$ yields a non-zero result.
	Moreover, $\HAM(S_1,\shft^m(S_2))\le k$ guarantees $S'_1 = \shft^m(S'_2)$,
	so the algorithm correctly ignores $m\in [n]$ if $\dec_0(\sk_0(S'_1),\sk_0(S'_2),m)$ yields a non-zero result.
	In the following, we assume $\sh(S'_1,S'_2)=0$ with $S'_1 = \shft^m(S'_2)$ for all the shifts considered.

	Recall that $\HAM(S_1,\shft^m(S_2))=\HAM(S_1,S'_1)+\HAM(S_2,S'_2)-|\PM_m|-|\PM'_m|$
	holds provided that $S'_1 = \shft^m(S'_2)$.
	Since $A\cap \rot^m(B)$ is a random subset of $[n]$ with elements sampled independently with rate $\frac{\log n}{\eps^2 k}$, the quantity $\frac{\eps^2 k}{\log n}(|\PM_m \cap A \cap \rot^m(B)|+|\PM'_m \cap A \cap \rot^m(B)|)$
	is with high probability a $\pm \frac{\eps k}{2}$-additive approximation of $|\PM_m|+|\PM'_m|$ (which can be argued as in the proof of Lemma~\ref{lem:chernoff}).
	Consequently, the computed value $d_m$ is with high probability a $\pm \frac{\eps k}{2}$-additive approximation of $\HAM(S_1,\shft^m(S_2))$.
	As $\sh(S_1,S_2)\le \HAM(S_1,S'_1)+\HAM(S_2,S'_2)$ holds by the triangle inequality,
	this means that the returned value is at least $(1-\eps)\sh(S_1,S_2)$ with high probability
	provided that $\sh(S_1,S_2)\ge \frac12k$.

	On the other hand, if $\sh(S_1,S_2)=\HAM(S_1,\shft^m(S_2))\le k$ for some shift $m\in [n]$,
	then $S'_1 = \shft^m(S'_2)$ and $d_m$	is a  $\pm \frac{\eps k}{2}$-additive approximation of $\sh(S_2,S_2)$.
	This either yields $\HAM(S_1,S'_1)+\HAM(S_2,S'_2) \le (1+\eps)\sh(S_1,S_2)$ (if $\PM_m \cap A \cap \rot^m(B) = \emptyset$, which yields $|\PM_m \cap A \cap \rot^m(B)|=|\PM'_m \cap A \cap \rot^m(B)|=0$)
	or that $m$ was among the shifts considered (otherwise).
	In both cases, we conclude that the returned value is at most $(1+\eps)\sh(S_1,S_2)$ with high probability.
\end{proof}

\subsection{Shift Distance Sketches for $\Sigma^n$.}
After handling non-pseudo-periodic and pseudo-periodic strings separately,
we derive sketches for the whole $\Sigma^n$.
The following results provide efficient shift distance decoding procedures
for the circular $k$-mismatch sketches described in Section~\ref{sec:summary}.

For the exact case, using Theorem~\ref{thm:exact-nap-shift-sketch} and Proposition~\ref{prp:exact-ap-shift-sketch},
the same construction as in the proof of Theorem~\ref{thm:exact-circular-sketch} yields the following corollary. 
\begin{corollary}[see Theorem~\ref{thm:exact-circular-sketch}]
There exists a $k$-ESDS sketch of size $\Ohtilde(k)$ with
decoding time $\Ohtilde(k^2)$.
\end{corollary}

For the approximate case, using Theorem~\ref{thm:approx-nap-shift-sketch} and Proposition~\ref{prp:approx-ap-shift-sketch},
the same construction as in Lemma~\ref{lem:approx-circular-sketch-relaxed} yields a relaxed $(\eps,k)$-ASDS sketch of size $\Ohtilde(\eps^{-2}\sqrt{k})$ with decoding time $\Ohtilde(\eps^{-2}k)$.

\begin{proposition}[see Lemma~\ref{lem:approx-circular-sketch-relaxed}]
There exists a relaxed $(\eps,k)$-ASDS sketch of size $\Ohtilde(\eps^{-2}\sqrt{k})$ and decoding time of $\Ohtilde(\eps^{-2}k)$.
\end{proposition}

The following  provides an alternative method for constructing $(\eps,k)$-ASDS sketches which improves the sketch size (for $k\ge \eps n$) but degrades the decoding time.

\begin{proposition}[see Property~\ref{prp:approx-circular-sketch-alternative}]\label{prp:approx-shift-sketch-alternative}
	There exists a decoding function which, together with the encoding $\circ_{\eps,k}$ of Definition~\ref{def:alternative},
	forms a relaxed $(\eps,k)$-ASDS sketch of $\Sigma^n$.
	The decoding time is $\Ohtilde(\frac{n^2}{\eps^2 k})$ with high probability.
\end{proposition}
\begin{proof}
	The decoding function, given the sketches of $S_1,S_2\in \Sigma^n$,
	computes a value\\ $d_m:=\HAM_{A \cap \rot^m(B)}(S_1,\shft^m(S_2))$ for each $m\in [n]$.
	For this, the algorithm iterates over $(i,S_1[i])$ with $i\in A$ (retrieved from the sketch of $S_1$)
	and $(i',S_2[i'])$ with $i'\in B$ (retrieved from the sketch of $S_2$),
	and increments  $d_{(i'-i)\cmod n}$ if $S_1[i]\ne S_2[i']$.

	As in the proof of Property~\ref{prp:approx-circular-sketch-alternative},
	$\frac{\eps^2 k}{\log n}\HAM_{A \cap \rot^m(B)}(S_1,\shft^m(S_2))$ is a $(1\pm \eps)$-approximation
	of $\HAM(S_1,\shft^m(S_2))$ with high probability provided that $\HAM(S_1,\shft^m(S_2))=\Omega(k)$ (and\\ $\frac{\eps^2 k}{\log n}\HAM_{A \cap \rot^m(B)}(S_1,\shft^m(S_2))=o(k)$ otherwise).
	Hence, the algorithm returns as an approximation of $\sh(S_1,S_2)$ the smallest value $\frac{\eps^2 k}{\log n}\HAM_{A \cap \rot^m(B)}(S_1,\shft^m(S_2))$ among $m\in [n]$.
\end{proof}

\begin{corollary}[see Theorem~\ref{thm:approx_circular_sketch}]
There exists an $(\eps,k)$-ASDS sketch of size $\Ohtilde(\eps^{-2}\sqrt{k})$ with decoding time $\Ohtilde(\eps^{-2}k)$, and an $(\eps,k)$-ASDS sketch of size $\Ohtilde(\eps^{-1.5}\sqrt{n})$ with decoding time $\Ohtilde(\eps^{-3}n)$.
\end{corollary}

\bibliography{bib}
\end{document}